\documentclass[lettersize,journal]{IEEEtran}
\IEEEoverridecommandlockouts

\usepackage{amsmath,amsfonts}
\usepackage{algorithm}
\usepackage{algpseudocode}
\usepackage{array}
\usepackage[caption=false,font=normalsize,labelfont=sf,textfont=sf]{subfig}
\usepackage{textcomp}
\usepackage{stfloats}
\usepackage{url}
\usepackage{verbatim}
\usepackage{graphicx}
\usepackage{cite}
\def\BibTeX{{\rm B\kern-.05em{\sc i\kern-.025em b}\kern-.08em
    T\kern-.1667em\lower.7ex\hbox{E}\kern-.125emX}}
\usepackage{balance}

\usepackage{makecell}
\usepackage{subfig}  
\usepackage{amssymb}
\usepackage{mathrsfs}
\usepackage{amsthm}
\RequirePackage{epsf}
\usepackage{epstopdf}
\usepackage[printonlyused]{acronym}

\acrodef{C-V2X}{cellular vehicle-to-everything}
\acrodef{V2I}{vehicle-to-infrastructure}
\acrodef{V2V}{vehicle-to-vehicle}
\acrodef{V2P}{vehicle-to-pedestrian}
\acrodef{RB}{resource block}
\acrodef{NR}{new radio}
\acrodef{QoS}{quality-of-service}
\acrodef{SINR}{signal-to-interference-plus-noise ratio}
\acrodef{CSI}{channel state information}
\acrodef{BS}{base station}
\acrodef{TxV}{transmitting vehicle}
\acrodef{RxV}{receiving vehicle}
\acrodef{RSSI}{received signal strength indicator}
\acrodef{RSS}{received signal strength}
\acrodef{PDF}{probability distribution function}
\acrodef{RV}{random variable}
\acrodef{MSE}{mean square error}
\acrodef{HR}{hazard rate}
\acrodef{ITS}{intelligent transportation systems}
\acrodef{PSFCH}{physical sidelink feedback channel}
\acrodef{PUCCH}{physical uplink control channel}
\acrodef{RV}{random variable}
\acrodef{i.i.d.}{independent, identically distributed}
\acrodef{GMM}{Gaussian mixture model}
\acrodef{TVaR}{tail value at risk}
\acrodef{CRF}{conditional relative frequency}
\acrodef{ORF}{overall relative frequency}
\acrodef{RSU}{roadside unit}
\acrodef{HD}{high-definition}
\acrodef{AV}{autonomous vehicle}
\acrodef{HPR}{high probability region}
\acrodef{CDF}{cumulative distribution function}
\acrodef{CCDF}{conditional cumulative distribution function}
\acrodef{AD}{autonomous driving}

\newtheorem{theorem}{Theorem}
\newtheorem{corollary}{Corollary}
\newtheorem{lemma}{Lemma}
\newtheorem{proposition}{Proposition}

\usepackage{siunitx}
\begin{document}

\title{Resilient Vehicular Communications under Imperfect Channel State Information \vspace{-0.2cm}}

\author{\normalsize Tingyu Shui, 
Walid Saad, \IEEEmembership{Fellow, IEEE}, Ye Hu, \IEEEmembership{Member, IEEE}, and Mingzhe Chen, \IEEEmembership{Senior Member, IEEE}\vspace{-0.8cm}
\thanks{T. Shui and W. Saad are with Bradley Department of Electrical and Computer Engineering, Virginia Tech, Alexandria, VA, 22305, USA, Emails: tygrady@vt.edu, walids@vt.edu.}
\thanks{Y. Hu is with the Department of Industrial and Systems Engineering, University of Miami, Coral Gables, FL, 33146, USA, Email: yehu@miami.edu.}
\thanks{M. Chen is with the Department of Electrical and Computer Engineering and Frost Institute for Data Science and Computing, University of Miami, Coral Gables, FL, 33146, USA, Email: mingzhe.chen@miami.edu.}
\thanks{A preliminary version of this work was presented at the IEEE International Conference on Communications \cite{shui2024resilienceperspectivecv2xcommunication}.}}
\maketitle
\begin{abstract}
\Ac{C-V2X} networks provide a promising solution to improve road safety and traffic efficiency. One key challenge in such systems lies in meeting \ac{QoS} requirements of vehicular communication links given limited network resources, particularly under imperfect \ac{CSI} conditions caused by the highly dynamic environment. In this paper, a novel two-phase framework is proposed to instill resilience into \ac{C-V2X} networks under unknown imperfect \ac{CSI}. The resilience of the \ac{C-V2X} network is defined, quantified, and optimized the first time through two principal dimensions: \emph{absorption phase} and \emph{adaptation phase}. Specifically, the \ac{PDF} of the imperfect \ac{CSI} is estimated during the absorption phase through dedicated absorption power scheme and \ac{RB} assignment. The estimated \ac{PDF} is further used to analyze the interplay and reveal the tradeoff between these two phases. Then, a novel metric named \textit{\ac{HR}} is exploited to balance the \ac{C-V2X} network's prioritization on absorption and adaptation. Finally, the estimated \ac{PDF} is exploited in the adaptation phase to recover the network's \ac{QoS} through a real-time power allocation optimization. Simulation results demonstrate the superior capability of the proposed framework in sustaining the \ac{QoS} of the \ac{C-V2X} network under imperfect \ac{CSI}. Specifically, in the adaptation phase, the proposed design reduces the \ac{V2V} delay that exceeds \ac{QoS} requirement by $35 \%$ and $56 \%$, and improves the average \ac{V2I} throughput by $14 \%$ and $16 \%$ compared to the model-based and data-driven benchmarks, respectively, without compromising the network's \ac{QoS} in the absorption phase.

\end{abstract}
\begin{IEEEkeywords}
\Ac{C-V2X}, Resilience, Imperfect \ac{CSI}, Power allocation.
\end{IEEEkeywords}

\acresetall
\vspace{-12pt}
\section{Introduction}
\Ac{C-V2X} networks are a key enabler for seamless inter-vehicular and vehicle-to-infrastructure communication \cite{9345798}.
By leveraging cellular \acp{BS}, \ac{C-V2X} facilitates emerging vehicular applications such as real-time \ac{HD} maps transmission for \ac{AD} \cite{9385411}, timely safety alerts \cite{9847235}, and inter-vehicular coordination in platoon \cite{8778746}. However, meeting the diverse \ac{QoS} requirements of vehicles in resource-constrained \ac{C-V2X} networks is challenging due to the need for accurate real-time \ac{CSI}. Specifically, the short coherence time, induced by dynamic multi-path effect and Doppler shift in vehicular networks, results in a mismatch between the obtained \ac{CSI} estimation and the actual channel conditions, which will ultimately lead to a \ac{QoS} deterioration in \ac{C-V2X} networks \cite{10197226}.

To ensure a desired \ac{QoS} in \ac{C-V2X} networks, accurate \ac{CSI} is essential for effective resource allocation. However, in realistic \ac{C-V2X} networks, the actual channel state may have changed by the time resource allocation based on estimated \ac{CSI} is applied, i.e., the imperfect \ac{CSI} problem. This mismatch between estimated \ac{CSI} and real-time channel, induced by multipath effects and Doppler shifts, is highly dependent on the wireless environment and the kinetic states of vehicles (e.g., velocity). Thus, \emph{model-based approaches} that assume a predefined error distribution or bounded error range, as commonly studied in the literature, are impractical. As a promising solution, \emph{data-driven approaches} require no prior knowledge of \ac{CSI} imperfection. However, most existing methods focus solely on achieving reliable network performance after the data-driven process, while neglecting the transient \ac{QoS} degradation during the process. This limitation is particularly critical in \ac{C-V2X} networks, where even slight \ac{QoS} degradation can lead to severe consequences \cite{9734746}.

To this end, there is a need for a framework that instills \emph{resilience} into the \ac{C-V2X} network, which considers both the in-progress and eventual system performance under arbitrary unknown imperfect \ac{CSI}. As an extended concept of reliability and robustness, resilience represents \textit{``the capability of a system to absorb the impact of unseen disruptions without prior information and finally adapt itself to the disruptions"}, addressing both of the two limitations in current research. Two key phases of resilience are \emph{absorption} and \emph{adaptation} \cite{reifert2024resiliencecriticalitybrothersarms}. “Absorption” is the system’s immediate reaction to the unseen disruption, such as switching operational modes or reconfiguring scheduling policies. These responses aim to ensure that the system continues to operate seamlessly and maintains the desired performance despite the disruption. During absorption, the system could also learn about the disruption from its effect, e.g., collecting data to accurately model the ongoing disruption or inferring its cause to localize and eliminate it. Following absorption, “adaptation” involves leveraging the acquired knowledge to mitigate the impact of the disruption and restore any degraded system performance. Ideally, the system is expected to recover to its original performance during adaptation, as if the disruption had not occurred. To this end, a resilient resource management framework for \ac{C-V2X} must therefore incorporate both phases to ensure sustained \ac{QoS} under the imperfect \ac{CSI} disruption.

\vspace{-10pt}
\subsection{Prior Works}
Although several reliable and robust resource management schemes have been studied to ensure desired \ac{QoS} under imperfect \ac{CSI} \cite{9400748, 8993812, 9374090, 9382930, 9857930, 10238756, 10213228}, the prior solutions may not be directly applicable to practical \ac{C-V2X} networks. For instance, prior works like \cite{9400748, 8993812, 9374090} assume that the imperfect \ac{CSI} error is either bounded by a known range or follows a certain distribution (e.g., Gaussian). In \cite{9400748} and \cite{8993812}, the authors studied the problem of power allocation and spectrum sharing for coexisting \ac{V2V} and \ac{V2I} links whereby the \ac{CSI} imperfection was assumed to follow a known distribution. Leveraging this prior knowledge, resources are allocated to meet \ac{QoS} requirements with high probability. In \cite{9374090}, the authors assumed that \ac{CSI} errors were bounded within a known range and incorporated a mapped fuzzy space to address uncertainty in joint time-frequency resource allocation. However, deploying these model-based approaches \cite{9400748, 8993812, 9374090} may lead to suboptimal performance in meeting desired \ac{QoS}, as accurately modeling imperfect \ac{CSI} is inherently challenging in highly dynamic vehicular networks. In the context of vehicular communication, the \ac{CSI} error characteristics can vary significantly across vehicular links and even within the same link under different traffic conditions. Thus, a predefined imperfect \ac{CSI} model that deviates from the actual channel state may further degrade the \ac{QoS}.

To overcome this limitation, data-driven approaches \cite{9382930, 9857930, 10238756, 10213228} have been proposed, eliminating the need for prior assumptions on \ac{CSI} errors. In general, these approaches rely on the design of robust resource allocation schemes that guarantee the worst-case \ac{QoS} by exploiting the information in collected imperfect \ac{CSI} samples. In \cite{9382930}, the authors exploited imperfect \ac{CSI} samples to construct a \ac{HPR} that probabilistically covers imperfect \ac{CSI} realizations, which was used to ensure a desired worst-case \ac{QoS}. Similarly, the work in \cite{9857930} employed a feasible region transformation method purely based on large-scale channel parameters. In \cite{10238756}, support vector clustering is used to capture \ac{CSI} error distribution in high-dimensional feature space, as a generalization of the \ac{HPR} method. In \cite{10213228}, a robust optimization problem based on the statistical characteristics of the \ac{CSI} (mean vector and covariance matrix) is formulated and solved. However, the solutions of \cite{9382930, 9857930, 10238756, 10213228} tend to be overly conservative while prioritizing worst-case \ac{QoS} at the expense of typical network performance, since the latent information of the \ac{CSI} error is not fully explored and exploited. Furthermore, none of the prior works \cite{9400748, 8993812, 9374090, 9382930, 9857930, 10238756, 10213228} considered the \ac{QoS} of \ac{C-V2X} networks during the data-driven process. From a resilience perspective, they only considered the \ac{QoS} in the adaptation phase while that in the absorption phase was largely ignored. 

\vspace{-8pt}
\subsection{Contributions}
The main contribution of this paper is a novel framework that defines, quantifies, and optimizes the resilience of a \ac{C-V2X} network in face of arbitrary unknown imperfect \ac{CSI}. Our goal is to meet the heterogeneous \ac{QoS} requirements of vehicular links by optimizing power allocation and \acp{RB} assignment. Specifically, we propose a resilience framework that can manage the impact of imperfect CSI across both \emph{absorption} and \emph{adaptation} phases. To our best knowledge, \textit{this is the first work that analyzes and optimizes the resilience of \ac{C-V2X} under imperfect \ac{CSI}, while, simultaneously considering the \ac{C-V2X}'s absorption and adaptation performance}. Our key contributions include:
\begin{itemize}
    \item We consider a \ac{C-V2X} network operating under arbitrary imperfect \ac{CSI} without prior information or assumptions. To meet the vehicular links' \ac{QoS} requirements, we formulate a bi-level optimization on the transmit power and \acp{RB} assignment. Due to the complexity of the bi-level structure and interdependence, we decouple the bi-level optimization into two sub-problems, solved sequentially in two phases: absorption and adaptation.
    \item We leverage the \ac{RSS} on vehicular links to estimate the \ac{PDF} of the \ac{CSI} error distribution in the absorption phase. Then, the \ac{MSE} between the true and estimated \ac{PDF} is defined as the \emph{adaptation capability} of the \ac{C-V2X} network. Based on the deconvolution estimation theory, an upper bound on the \ac{MSE} is derived and minimized by optimizing the \acp{RB} assignment and a dedicated absorption power scheme. Moreover, the derived upper bound shows a tradeoff between a high adaptation capability and a compromised network \ac{QoS} performance in the absorption phase. Consequently, we further incorporate a novel metric named \textit{\ac{HR}} to evaluate the severity of \ac{QoS} degradation in absorption phase. 
    \item After absorption, power schemes in the adaptation phase are optimized based on real-time imperfect \ac{CSI}. Using the estimated \ac{PDF}, the probability of satisfying the \ac{QoS} requirement is derived. Then, the original non-convex problem is approximated and solved by a one-dimensional search algorithm. Moreover, we analytically show that a satisfactory \ac{QoS} is jointly influenced by the \ac{C-V2X} network's absorption to the imperfect \ac{CSI} and the quality of the real-time \ac{CSI} in adaptation. 
    \item Simulation results show the proposed framework outperforms model-based and data-based benchmarks across both phases. Particularly, the proposed design reduces the \ac{V2V} delay when the delay exceeds the desired requirement by $35 \%$ and $56 \%$ and improves the average \ac{V2I} throughput by $14 \%$ and $16 \%$ the model-based and data-driven benchmarks, respectively, without compromising the network's \ac{QoS} in the absorption phase.
\end{itemize}

\vspace{-10pt}
\section{System Model}
Consider a \ac{C-V2X} network consisting of a set $\mathcal{N}$ of $N$ \ac{V2I} links. A centralized \ac{RSU} allocates orthogonal \acp{RB} to each \ac{V2I} link for cellular uplink and downlink transmission through the Uu interface \cite{9345798}. Within the coverage of the \ac{RSU}, a set $\mathcal{M}$ of $M$ \ac{V2V} links that use \ac{NR} \ac{C-V2X} transmission mode-1 through \ac{NR} sidelinks are deployed. Typically, the \ac{V2V} links will be allocated dedicated \acp{RB} to transmit time-sensitive and safety-critical messages. However, this \ac{RB} partitioning design may be inefficient since dedicated \acp{RB} will be limited for emergent situations, e.g., a sudden surge in \ac{V2V} transmission. Thus, we consider a more flexible utilization of \acp{RB} in which \ac{V2V} links will reuse the \ac{RB}s allocated to the \ac{V2I} links during the \ac{V2I} uplink transmission \cite{7913583}, as shown in Fig. \ref{System Model}. In our system, a \ac{V2V} link will only reuse a single \ac{RB} and the allocated \ac{RB} of a given \ac{V2I} link can be only shared with one \ac{V2V} link to mitigate interference \cite{8993812, 9400748, 9374090, 9382930, 9857930}. In practice, the number of \ac{V2I} links $N$ is generally no smaller than the number of \ac{V2V} links $M$ \cite{9400748}. Thus, we consider the worst-case scenario $M = N$. Under \ac{NR} \ac{C-V2X} transmission mode 1 \cite{9530506}, a centralized \ac{RSU} will optimize the matching between \ac{V2V} links and \ac{V2I} links for sharing the same \ac{RB} and allocate the transmit power of these links to satisfy the desired \ac{QoS} requirements. 
\begin{figure}[t]
	\centering
	\includegraphics[width=0.33\textwidth]{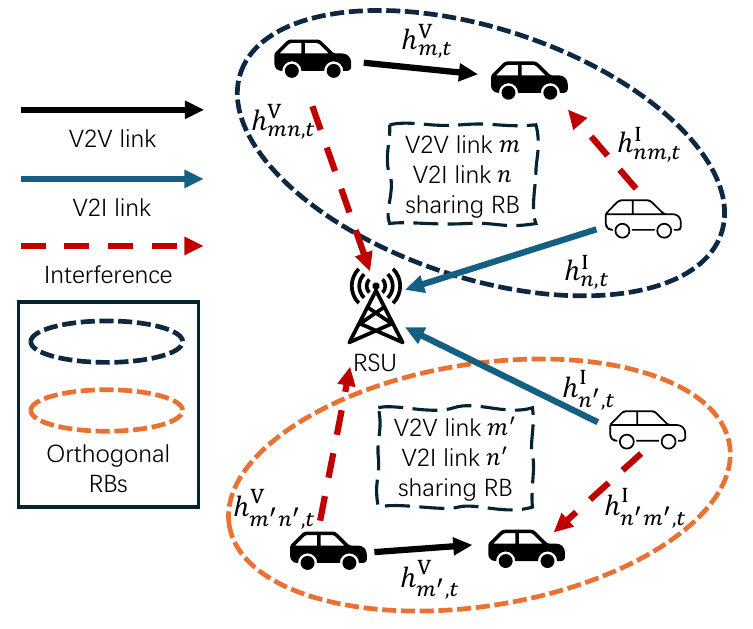}
    \vspace{-6pt}
	\caption{\small{System model of the considered \ac{C-V2X} network. Orthogonal \acp{RB} are shared by vehicular links with heterogeneous \ac{QoS} requirements.}}
 \label{System Model}
 \vspace{-15pt}
\end{figure}

\vspace{-10pt}
\subsection{Vehicular \ac{QoS} requirements}
We assume that the system operates in a time-slotted manner. Each time slot is equal to the coherence time of the channel during which the \ac{CSI} is invariant. We define $p^{\textrm{I}}_{n,t}$ and $p^{\textrm{V}}_{m,t}$ as, respectively, the transmit power of \ac{V2I} link $n$ and \ac{V2V} link $m$ at time slot $t$. The \ac{SINR} over \ac{V2I} link $n$ can be obtained as $\gamma^{\textrm{I}}_{n,t}= \frac{ p^{\textrm{I}}_{n,t} h^{\textrm{I}}_{n,t} }{ \sum_{m=1}^{M} \alpha_{mn} p^{\textrm{V}}_{m,t} h^{\textrm{V}}_{mn,t} + \sigma^2 }$, where $h^{\textrm{I}}_{n,t}$ is the channel gain on \ac{V2I} link $n$, $h^{\textrm{V}}_{mn,t}$ is the channel gain on the interference link from \ac{V2V} link $m$ to \ac{V2I} link $n$, and $\sigma^2$ is the power of additive white Gaussian noise. Moreover, $\alpha_{mn}$ is a binary variable with $\alpha_{mn}=1$ indicating that \ac{V2I} link $n$ is sharing its \ac{RB} with \ac{V2V} link $m$. Here, we do not consider the time-varying nature of $\alpha_{mn}$ since frequently changing the link matching would introduce additional delays, increase system overhead, and potentially degrade \ac{QoS} due to handover. Consequently, the matching process varies at a time scale that is generally vary much slower than the \ac{CSI} and, thus, we focus on a time interval of $C$ time slots where all matching are fixed. Similarly, we can obtain the \ac{SINR} over \ac{V2V} link $m$ as $\gamma^{\textrm{V}}_{m,t}= \frac{ p^{\textrm{V}}_{m,t} h^{\textrm{V}}_{m,t} }{ \sum_{n=1}^{N} \alpha_{mn} p^{\textrm{I}}_{n,t} h^{\textrm{I}}_{nm,t} + \sigma^2 }$, where $h^{\textrm{V}}_{m,t}$ and $h^{\textrm{I}}_{nm,t}$ represent the channel gain on \ac{V2V} link $m$ and the interference link from \ac{V2I} link $n$ to \ac{V2V} link $m$. 

We define $\boldsymbol{p}_{\mathcal{M},t} = \left[ p^{\textrm{V}}_{1,t}, \ldots, p^{\textrm{V}}_{M,t} \right]$, $\boldsymbol{p}_{\mathcal{N},t} = \left[ p^{\textrm{I}}_{1,t}, \ldots, p^{\textrm{I}}_{M,t} \right]$, and $\boldsymbol{A} = \left[\boldsymbol{\alpha}_{1}, \ldots, \boldsymbol{\alpha}_{N}\right]$ with $\boldsymbol{\alpha}_{n} = \left[ \alpha_{1n}, \ldots, \alpha_{Mn}\right]^T, \forall n \in \mathcal{N}$. Note that $\boldsymbol{p}_{\mathcal{M},t}$ and $\boldsymbol{p}_{\mathcal{N},t}$ are optimized in real time while $\boldsymbol{A}$ is fixed and should be determined at time slot $t = 1$. By further defining $\boldsymbol{h}^{nm}_{t} = \left[ h^{\textrm{I}}_{n,t}, h^{\textrm{I}}_{nm,t}, h^{\textrm{V}}_{m,t}, h^{\textrm{V}}_{mn,t} \right]$ and $\mathcal{H}_{t} = \{ \boldsymbol{h}^{nm}_{t} \mid n \in \mathcal{N}, m \in \mathcal{M} \}$ as the overall \ac{CSI} set at $t$, the dynamics of $\boldsymbol{p}_{\mathcal{M},t}$, $\boldsymbol{p}_{\mathcal{N},t}$, and $\boldsymbol{A}$ over different times scale are illustrated in Fig.~\ref{timescale_match_power}. Aligned with the literature \cite{8993812, 9400748, 9374090, 9382930, 9857930}, different \ac{QoS} requirements over vehicular links are considered, i.e., capacity on \ac{V2I} links and latency on \ac{V2V} links. Thus, the \ac{QoS} requirements on \ac{V2I} link $n$ and \ac{V2V} link $m$ at time slot $t$ will be given by:
\begin{equation}
\label{throughput}
    R_{n,t} \triangleq B \log (1+\gamma^{\textrm{I}}_{n,t}) \geq R_0, \forall n \in \mathcal{N},
\end{equation}
\vspace{-10pt}
\begin{equation}
\label{delay}
   \tau_{m,t} \triangleq \frac{D}{B \log (1+\gamma^{\textrm{V}}_{m,t})} \leq \tau_0, \forall m \in \mathcal{M},
\end{equation}
where $B$ is the bandwidth of each \ac{RB}, $D$ is the \ac{V2V} link packet size, and $R_0$ and $\tau_0$ are respectively the given throughput and delay requirements. 

\vspace{-8pt}
\subsection{Imperfect \ac{CSI} Model}
For all vehicular links, the channel gain is composed of large-scale and small-scale fading. Without loss of generality, we consider an arbitrary component $h_t$ of $\boldsymbol{h}^{nm}_{t}$, modeled as $h_t = L_t |g_t|^2$, where $L_t$ and $g_t$ represent the large-scale and small-scale fading components, respectively \cite{9400748, 8993812, 9374090, 9382930, 9857930, 10238756, 10213228}. Specifically, based on the free space propagation path-loss model, the large-scale fading is modeled as $L_t = G_t \zeta d_t^{-\alpha}$ with path loss gain $G_t$, log-normal shadow fading gain $\zeta$, path loss exponent $\alpha$, and the link distance $d_t$ of this vehicular link. Due to the severe multipath effects in urban vehicular communication environments, the small-scale fading is modeled as Rayleigh fading and represented by $g_t \sim \mathcal{C N}(0,1)$. 

Our goal is to determine $\boldsymbol{p}_{\mathcal{M},t}$, $\boldsymbol{p}_{\mathcal{N},t}$, and $\boldsymbol{A} $ that meet the \ac{QoS} requirements in \eqref{throughput} and \eqref{delay}. The primary challenge of this problem lies in accurately obtaining the \ac{CSI} in $\mathcal{H}_t$ in vehicular network. Specifically, in \ac{NR} \ac{C-V2X} mode 1, all involved vehicular links are required to report their \ac{CSI} to the \ac{RSU}, either directly through \ac{PUCCH} or relayed by the \ac{PSFCH} \cite{8998153}. However, as a consequence of the Doppler shift and multipath effects, the coherence time in dynamic vehicular network is very short. Thus, the \ac{CSI} fed back to the \ac{RSU} may become outdated due to the delay introduced in \ac{CSI} feedback relay and establishment of \ac{PUCCH} and \ac{PSFCH}. In other words, when the designed power $\boldsymbol{p}_{\mathcal{M},t}$ and $\boldsymbol{p}_{\mathcal{N},t}$ are deployed, the true channel will inevitably deviate from the \ac{CSI} feedback received by the \ac{RSU}, which induces the imperfect \ac{CSI} problem\footnote{Such \ac{CSI} imperfection cannot be mitigated through more advanced estimation algorithms since it is fundamentally caused by the short coherence time and rapid channel variations.} in vehicular network \cite{9400748, 8993812, 9374090, 9382930, 9857930, 10238756, 10213228}. As a result, $\boldsymbol{p}_{\mathcal{M},t}$ and $\boldsymbol{p}_{\mathcal{N},t}$ are actually determined according to the imperfect \ac{CSI} feedback $\mathcal{\hat{H}}_t = \{ \hat{\boldsymbol{h}}^{nm}_{t} \mid n \in \mathcal{N}, m \in \mathcal{M} \}$.
\begin{figure}[t]
\captionsetup[subfigure]{font=small}  
    \centering
    \subfloat[]{
        \includegraphics[width=0.42\textwidth]{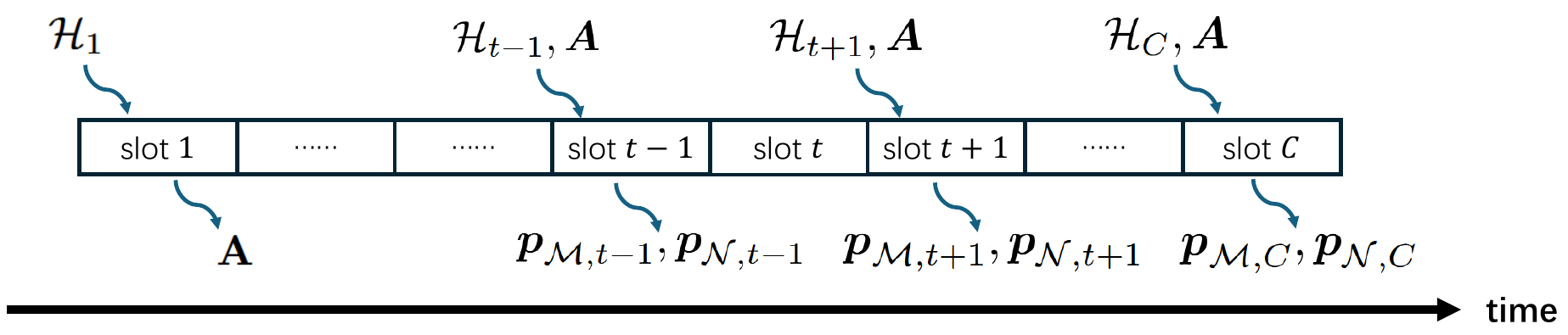}
        \label{timescale_match_power}
    }
    \\
    \vspace{-12pt}
    \subfloat[]{
        \includegraphics[width=0.42\textwidth]{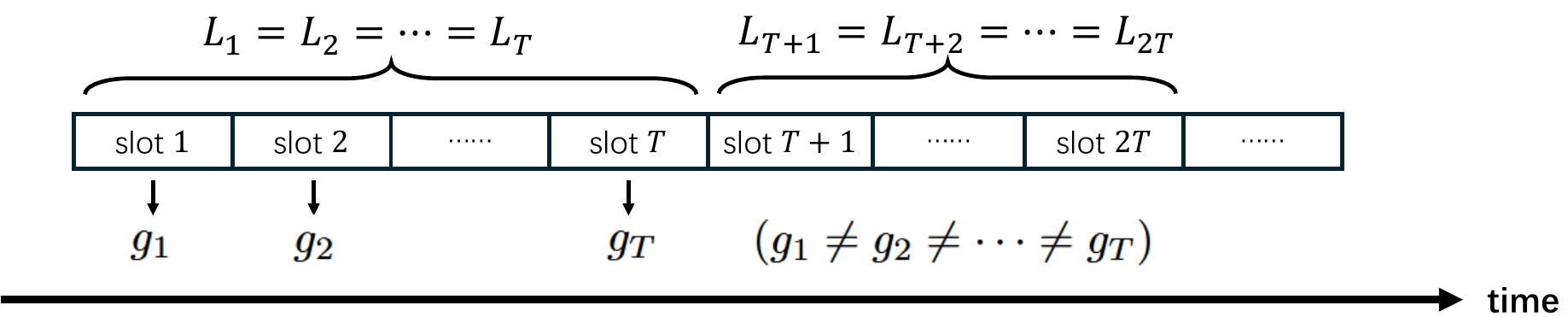}
        \label{timescale_channel}
    }
        \vspace{-6pt}
    \caption{Illustration of how the various parameters change over time scales: a) matching $\boldsymbol{A} $ and transmit power scheme $\boldsymbol{p}_{\mathcal{M},t}$, $\boldsymbol{p}_{\mathcal{N},t}$, and b) large-scale fading $L_t$ and small-scale fading $g_t$.}
    \vspace{-15pt}
\end{figure}

Since the large-scale fading $L_t$ is dominantly determined by the locations of the vehicles and varies on a slow timescale compared to small-scale fading \cite{7270335}, we assume that $L_t$ can be perfectly estimated by the \ac{RSU} and remains constant over $T$ time slots \cite{9741813}, as shown in Fig. \ref{timescale_channel}. Therefore, the imperfection in \ac{CSI} dominantly stems from small-scale fading. Take $\hat{\boldsymbol{h}}^{nm}_{t}$ and $\boldsymbol{h}^{nm}_{t}$ as an example, we assume $\hat{g}^{\textrm{I}}_{n,t} = g^{\textrm{I}}_{n,t}$ and $\hat{g}^{\textrm{V}}_{mn,t} = g^{\textrm{V}}_{mn,t}$, while $\hat{g}^{\textrm{V}}_{m,t} \neq g^{\textrm{V}}_{m,t}$ and $\hat{g}^{\textrm{I}}_{nm,t} \neq g^{\textrm{I}}_{nm,t}$. The reason for assuming $\hat{g}^{\textrm{I}}_{n,t} = g^{\textrm{I}}_{n,t}$ and $\hat{g}^{\textrm{V}}_{mn,t} = g^{\textrm{V}}_{mn,t}$ is that $\hat{g}^{\textrm{I}}_{n,t}$ and $\hat{g}^{\textrm{V}}_{mn,t}$ are directly estimated by the \ac{RSU} without relaying. However, $\hat{g}^{\textrm{V}}_{m,t}$ and $\hat{g}^{\textrm{I}}_{nm,t}$ are both estimated by the transmitting vehicle on \ac{V2V} links and relayed to the \ac{RSU} through \ac{PSFCH}, which results in severe \ac{CSI} feedback delay and \ac{CSI} imperfection. Since \ac{V2V} links are generally established when the relative movement of two vehicles are comparatively stable and predictable \cite{8993812}, we first model $\hat{g}^{\textrm{V}}_{m,t}$ and $g^{\textrm{V}}_{m,t}$ directly using a first-order Gauss-Markov process:
\begin{equation}
\label{i v2v}
   | g^{\textrm{V}}_{m,t} | ^2 = | \hat{g}^{\textrm{V}}_{m,t} |^2 + (1 - \delta_{m,t}^2) \left( | {e}_{m,t} |^2 - | \hat{g}^{\textrm{V}}_{m,t} |^2 \right),
\end{equation}
where $\delta_{m,t} = J_0\left(2 \pi f_D \Delta_t\right)$ is the coefficient given by Jakes statistical model \cite{7913583} and the error term ${e}_{m,t} \sim \mathcal{C N}(0,1)$ is \ac{i.i.d.} across different time slots. Specifically, $J_0$ is the zero-order Bessel function of the first kind, $\Delta_t$ is the \ac{CSI} feedback delay, and $f_D = \frac{v f_c}{c}$ is the maximum Doppler frequency with $c$ being the speed of light, where $v$ and $f_c$ are the vehicle speed and carrier frequency respectively. In general, $\delta_{m,t}$ is available to \ac{V2V} link $m$ \cite{7913583, 9400748, 8993812}, as $\Delta_t$ can be inferred from the timestamp used in \ac{CSI} estimation. However, it is almost impossible to find a suitable model for $\hat{g}^{\textrm{I}}_{nm,t}$ and $g^{\textrm{I}}_{nm,t}$ on the interference link, since the relative movement between \ac{V2V} link $m$ and \ac{V2I} link $n$ is highly dynamic and hard to model.
Thus, we only assume an additive error ${e}_{nm,t}$, where ${e}_{nm,t} \sim \mathcal{E}_{nm}$ is an \ac{i.i.d.} \ac{RV} across different time slots of unknown distribution $\mathcal{E}_{nm}$. Then, $\hat{g}^{\textrm{I}}_{nm,t}$ and ${g}^{\textrm{I}}_{nm,t}$ are modeled as:
\begin{equation}
\label{i v2i}
   | g^{\textrm{I}}_{nm,t} | ^2 = | \hat{g}^{\textrm{I}}_{nm,t} |^2 + {e}_{nm,t}.
\end{equation}

\vspace{-12pt}
\section{Problem Formulation}
Due to the presence of stochastic \ac{CSI} error ${e}_{m,t} \sim \mathcal{C N}(0,1)$ in \eqref{i v2v} and ${e}_{nm,t} \sim \mathcal{E}_{nm}$ in \eqref{i v2i}, we can only seek to satisfy \eqref{delay} in a probabilistic manner. Specifically, the \ac{V2V} link's \ac{QoS} requirement is redefined to ensure that the probability of \eqref{delay} is at least $P_0$, i.e., $\mathbb{P} \{ \tau_{m,t} \leq \tau_0 \} \geq P_0, \forall m \in \mathcal{M}$. Based on \eqref{delay}, \eqref{i v2v}, and \eqref{i v2i}, we can further derive $\mathbb{P} \{ \tau_{m,t} \leq \tau_0 \}$ as
\begin{equation}
\small
\label{prob}
\begin{aligned}
    & P_{m,t}(\mathcal{E}_{m}, \boldsymbol{p}_{\mathcal{M},t}, \boldsymbol{p}_{\mathcal{N},t}, \boldsymbol{A} ) \triangleq \mathbb{P} \{ \tau_{m,t} \leq \tau_0 \} \\
    = & \mathbb{P} \Biggl\{ \frac{p^{\textrm{V}}_{m,t} L^{\textrm{V}}_{m,t}}{ \gamma_{\textrm{V}} } (1 - \delta_{m,t}^2) | {e}_{m,t} |^2 \\
    & \quad - \sum_{n=1}^{N} \alpha_{mn} p^{\textrm{I}}_{n,t} L^{\textrm{I}}_{nm,t} {e}_{nm,t} \geq b_{m,t} \Biggm| {e}_{nm,t} \sim \mathcal{E}_{nm}, n \in \mathcal{N} \Biggr\},
\end{aligned}
\end{equation}
where we define $\mathcal{E}_{m} \triangleq \left\{ \mathcal{E}_{1m}, \ldots, \mathcal{E}_{Nm} \right\}$, $b_{m,t} = \sigma^2 + \sum_{n=1}^{N} \alpha_{mn} p^{\textrm{V}}_{n,t} L^{\textrm{V}}_{nm,t} | \hat{g}^{\textrm{V}}_{nm,t} | ^2 - \frac{p^{\textrm{I}}_{m,t} L^{\textrm{I}}_{m,t} \delta_{m,t}^2 | \hat{g}^{\textrm{I}}_{m,t} | ^2}{\gamma_{\textrm{V}}}$ is determined by the imperfect \ac{CSI} $\hat{g}^{\textrm{I}}_{m,t}$ and $\hat{g}^{\textrm{V}}_{nm,t}$, and $\gamma_{\textrm{V}} = 2^{\frac{D}{B\tau_0}}-1$ is a constant. However, \eqref{prob} is still intractable since the error distribution $\mathcal{E}_{nm}$ is unknown. Moreover, it is impractical to assume a certain prior distribution on $\mathcal{E}_{nm}$. 

To design a resilient \ac{C-V2X} network under the distribution of imperfect \ac{CSI}, we propose a two-phase framework next. Specifically, the network first estimates the error distribution $\mathcal{E}_{nm}$ through a dedicated \emph{absorption phase} and then recovers the \ac{QoS} in the \emph{adaptation phase} exploiting the estimated result $\hat{\mathcal{E}}_{m} \triangleq \{ \hat{\mathcal{E}}_{1m}, \ldots, \hat{\mathcal{E}}_{Nm} \}$. Unlike prior designs that focus solely on reliability or robustness and ensure network performance only during adaptation, our proposed framework simultaneously considers absorption performance and analytically evaluates the impact of absorption strategies on network's adaptation performance. The proposed framework instills resilience into the \ac{C-V2X} network under \ac{CSI} imperfection without any prior information or external intervention, as explained in Sections \ref{absorption section} and \ref{adaptation section}. 

The goal of the proposed two-phase resilient framework is to ensure that the probability in \eqref{prob}, computed using the estimated error distribution $\hat{{\mathcal{E}}}_{m}$, closely approximates its true counterpart. This prevents significant \ac{QoS} degradation on \ac{V2V} links under the imperfect \ac{CSI} disruption. To this end, we first define the overall deviation at time slot $t$ as
\begin{equation}
\vspace{-5pt}
\label{MSE obj}
\begin{aligned}
    & J_t(\boldsymbol{p}_{\mathcal{M},t}, \boldsymbol{p}_{\mathcal{N},t}, \boldsymbol{A} ) \\
    = & \sum_{m\in\mathcal{M}} \left[P_{m,t}(\hat{\mathcal{E}}_{m}, \boldsymbol{p}_{\mathcal{M},t}, \boldsymbol{p}_{\mathcal{N},t}, \boldsymbol{A} ) - P_{m,t}(\mathcal{E}_{m}, \boldsymbol{p}_{\mathcal{M},t}, \boldsymbol{p}_{\mathcal{N},t}, \boldsymbol{A} ) \right]^2.
\end{aligned}
\end{equation}
Then, we can formulate a bi-level optimization problem given in \eqref{opt1}, 
\begin{figure*}
\vspace{-20pt}
\begin{subequations}
\label{opt1}
\begin{IEEEeqnarray}{s,rCl'rCl'rCl}
& \underset{\boldsymbol{A} }{\text{min}} &\quad& \sum_{t=1}^{C} J_t(\boldsymbol{p}^{*}_{\mathcal{M},t}, \boldsymbol{p}^{*}_{\mathcal{N},t}, \boldsymbol{A} )\label{obj1}\\
&\text{s.t.} && \sum_{m \in \mathcal{M}} \alpha_{mn} = 1, \forall n \in \mathcal{N}, \sum_{n \in \mathcal{N}} \alpha_{mn} = 1, \forall m \in \mathcal{M}  \label{c1-1},\\
&&& (\boldsymbol{p}^{*}_{\mathcal{M},t}, \boldsymbol{p}^{*}_{\mathcal{N},t}) = \arg \min_{(\boldsymbol{p}_{\mathcal{M},t},  \boldsymbol{p}_{\mathcal{N},t}) \in \mathcal{G}_t (\boldsymbol{A} )} J_t(\boldsymbol{p}_{\mathcal{M},t}, \boldsymbol{p}_{\mathcal{N},t}, \boldsymbol{A} ), \ t = 1, \ldots, C   \label{c1-2}.
\end{IEEEeqnarray}
\end{subequations}
\vspace{-25pt}
\end{figure*}
where \eqref{obj1} is the objective function of the upper-level problem considering \eqref{MSE obj} over the whole $C$ time slots, \eqref{c1-1} is the matching constraint on $\boldsymbol{A} $, and \eqref{c1-2} is the lower-level problem at time slot $t$, with feasible region $\mathcal{G}_t(\boldsymbol{A} )$ given by:
\begin{equation}
\small
\vspace{-5pt}
\label{feasible}
\begin{aligned}
&\mathcal{G}_t (\boldsymbol{A} ) \\
= & \left\{ (\boldsymbol{p}_{\mathcal{M},t}, \boldsymbol{p}_{\mathcal{N},t}) \middle|
\begin{array}{l}
R_{n,t}( \boldsymbol{p}_{\mathcal{M},t}, \boldsymbol{p}_{\mathcal{N},t}, \boldsymbol{A} ) \geq R_0, \\
P_{m,t}(\hat{\mathcal{E}}_{m}, \boldsymbol{p}_{\mathcal{M},t}, \boldsymbol{p}_{\mathcal{N},t}, \boldsymbol{A} )\geq P_0,\\
p^{\textrm{V}}_{\textrm{min}} \leq p^{\textrm{V}}_{m,t} \leq p^{\textrm{V}}_{\textrm{max}}, \\
p^{\textrm{I}}_{\textrm{min}} \leq p^{\textrm{I}}_{n,t} \leq p^{\textrm{I}}_{\textrm{max}}, \forall n \in \mathcal{N}, \forall m \in \mathcal{M}
\end{array} 
\right\}. 
\end{aligned}
\end{equation}
In \eqref{feasible}, the vehicular links' requirements on \ac{QoS}, minimum transmit power, and maximum transmit power are considered. To make \eqref{feasible} more precise and rigorous, we rewrite $R_{n,t}$ in \eqref{throughput} as $R_{n,t}( \boldsymbol{p}_{\mathcal{M},t}, \boldsymbol{p}_{\mathcal{N},t}, \boldsymbol{A} )$. Notably, the constraint $P_{m,t}(\hat{\mathcal{E}}_{m}, \boldsymbol{p}_{\mathcal{M},t}, \boldsymbol{p}_{\mathcal{N},t}, \boldsymbol{A} ) \geq P_0$ in \eqref{feasible} represents the probability based on estimated $\hat{\mathcal{E}}_{m}$ other than $\mathcal{E}_{m}$. 
To elaborate further on the bi-level problem \eqref{opt1}, we first focus on the lower-level problem. In \eqref{c1-2}, $\boldsymbol{p}_{\mathcal{M},t}$ and $ \boldsymbol{p}_{\mathcal{N},t}$ are optimized in each time slot $t$ based on the real-time \ac{CSI} to meet the desired \ac{QoS}, as in \eqref{feasible}. However, the time-invariant matching $\boldsymbol{A}$ is determined by the upper-level problem of \eqref{c1-1}. In other words, the upper-level problem is required to find the optimal $\boldsymbol{A} $ over an extended period without knowing $\boldsymbol{p}^{*}_{\mathcal{M},t}$, $\boldsymbol{p}^{*}_{\mathcal{N},t}$, and $ J_t(\boldsymbol{p}^{*}_{\mathcal{M},t}, \boldsymbol{p}^{*}_{\mathcal{N},t}, \boldsymbol{A} )$ in the future, as captured by the summation in \eqref{obj1}. 

Solving the bi-level problem \eqref{opt1} presents three challenges: (i) The true imperfection distribution $\mathcal{E}_{m}$ is unknown for both the upper-level and lower-level problems; (ii) For the upper-level problem, the objective function \eqref{obj1} is not tractable since we have no access to the \ac{CSI} in the future. In other words, we need to determine $\boldsymbol{A} $ at $t=1$ without knowing $\mathcal{\hat{H}}_t$, $\boldsymbol{p}^{*}_{\mathcal{M},t}$, and $\boldsymbol{p}^{*}_{\mathcal{N},t}$ for $\forall t > 1$; (iii) The coupling of upper-level and lower-level problems introduces a nested structure and interdependence, which complicates the optimization \eqref{opt1}. Specifically, the optimal matching $\boldsymbol{A} $ depends on the transmit powers $\boldsymbol{p}_{\mathcal{M},t}$ and $\boldsymbol{p}_{\mathcal{N},t}$ in each time slot while the optimal $\boldsymbol{p}_{\mathcal{M},t}$ and $\boldsymbol{p}_{\mathcal{N},t}$ depend, in turn, on the optimal matching $\boldsymbol{A} $. 

To address these challenges, we aim at a sub-optimal solution by decoupling \eqref{opt1} into two sub-problems that are solved sequentially in the aforementioned two-phase framework, as shown in Fig. \ref{timescale_two_phase}. Specifically, the matching variable $\boldsymbol{A} $ is determined together with a dedicated absorption power scheme $\boldsymbol{p}_{\mathcal{M}, \textrm{a}} = \left[ p_{1, \textrm{a}}^{\textrm{V}}, \ldots, p_{M, \textrm{a}}^{\textrm{V}} \right]$ and $\boldsymbol{p}_{\mathcal{N}, \textrm{a}} = \left[ p_{1, \textrm{a}}^{\textrm{I}}, \ldots, p_{N, \textrm{a}}^{\textrm{I}} \right]$ in the \emph{absorption phase} for maintaining network \ac{QoS} and accurately estimating $\mathcal{E} \triangleq \{ \mathcal{E}_1, \ldots, \mathcal{E}_M \}$, without need for future \ac{CSI}. Subsequently, $\boldsymbol{p}_{\mathcal{M},t}$ and $\boldsymbol{p}_{\mathcal{N},t}$ are optimized, based on the estimation result $\hat{\mathcal{E}} \triangleq \{ \hat{\mathcal{E}}_1, \ldots, \hat{\mathcal{E}}_M \}$ and real-time \ac{CSI} during the \emph{adaptation phase}, to recover the \ac{QoS} of vehicular links. Different from the reliable and robust designs widely explored in recent literature \cite{9400748, 8993812, 9374090, 9382930, 9857930, 10238756, 10213228}, the proposed framework solves problem \eqref{opt1} from a resilience perspective by considering the interplay of three key processes: maintaining network \ac{QoS}, accurately estimating \ac{CSI} imperfection, and effectively leveraging the estimated result to recover network \ac{QoS}.

\vspace{-10pt}
\section{Absorption Phase}
\label{absorption section}
We now discuss the \emph{absorption phase} where the \ac{PDF} of error distribution is estimated by the \ac{RSU}. Specifically, we derive an analytical upper bound on the \ac{MSE} between the estimated and true \ac{PDF}. From a resilience perspective, this upper bound is defined as the \emph{adaptation capability} of the \ac{C-V2X} network and minimized in the upper-level problem of \eqref{opt1}. Based on our analysis, we also show a tradeoff between the communication \ac{QoS} in the absorption phase and the \ac{C-V2X} network's adaptation capability, which is captured in the optimization via a novel metric named \emph{hazard rate} (HR).
\vspace{-8pt}
\subsection{Deconvolution based Estimation}
To estimate the \ac{PDF} of error distribution and sustain the \ac{QoS} of the \ac{C-V2X} network, the \ac{RSU} will switch to a dedicated absorption phase lasting for $T$ time slots, as shown in Fig. \ref{timescale_two_phase}. Our goal is to define and find the optimal absorption power scheme $\boldsymbol{p}_{\mathcal{M}, \textrm{a}}$, $\boldsymbol{p}_{\mathcal{N}, \textrm{a}}$, and the matching variable $\boldsymbol{A}$. Since $\boldsymbol{A} $ should satisfy \eqref{c1-1}, we first analyze the optimal power scheme of a general case in which \ac{V2I} link $n$ is sharing its \ac{RB} with \ac{V2V} link $m$, i.e., $\alpha_{mn} = 1$. We assume that the $T$ time slots align with the interval where the large-scale fading parameters $L^{\textrm{V}}_{m, \textrm{a}}$ and $L^{\textrm{I}}_{nm, \textrm{a}}$, as well as the channel coefficient $\delta_m$ in \eqref{i v2v} are invariant.

Next, we can leverage the \ac{RSS} at the receiving vehicle of \ac{V2V} link $m$ for estimation. Particularly, the true \ac{RSS} at the receiving vehicle of \ac{V2V} link $m$ at time slot $k$ can be given as ${r}_{m,k} = p^{\textrm{I}}_{n, \textrm{a}} L^{\textrm{I}}_{nm, \textrm{a}} | g^{\textrm{I}}_{nm,k} | ^2 + p^{\textrm{V}}_{m,\textrm{a}} L^{\textrm{V}}_{m,\textrm{a}} | g^{\textrm{V}}_{m,k} | ^2 + \sigma^2$. The true \ac{RSS} ${r}_{m,k}$ will be fed back to the \ac{RSU} through \ac{PUCCH} and \ac{PSFCH}, forming a true \ac{RSS} set ${\mathcal{R}_m} = \left\{ {r}_{m,1}, \ldots, {r}_{m,T} \right\}$. Correspondingly, the \ac{RSU} can directly calculate the nominal \ac{RSS} at time slot $k$ based on $\hat{\mathcal{H}}_k$, which is given by $\hat{r}_{m,k} = p^{\textrm{I}}_{n,\textrm{a}} L^{\textrm{I}}_{nm,\textrm{a}} | \hat{g}^{\textrm{I}}_{nm,k} | ^2 + p^{\textrm{V}}_{m, \textrm{a}} L^{\textrm{V}}_{m,\textrm{a}} | \hat{g}^{\textrm{V}}_{m,k} | ^2 + \sigma^2$. Thus, the \ac{RSU} can form a nominal \ac{RSS} set $\hat{\mathcal{R}}_{m} = \left\{ \hat{r}_{m,1}, \ldots, \hat{r}_{m,T} \right\}$ at the end of the absorption phase. Then, the \ac{RSU} can collect a sequence of data samples $\mathcal{Z}_m = \left\{ z_{m,1}, \ldots, z_{m,T} \right\}$ with $z_{m,k}$ defined as
\begin{equation}
\label{data sample}
\begin{aligned}
   z_{m,k}  
   & \triangleq \frac{{r}_{m,k} - \hat{r}_{m,k}}{p^{\textrm{I}}_{n,\textrm{a}} L^{\textrm{I}}_{nm,\textrm{a}}} + \frac{p^{\textrm{V}}_{m,\textrm{a}} L^{\textrm{V}}_{m,\textrm{a}} }{p^{\textrm{I}}_{n,\textrm{a}} L^{\textrm{I}}_{nm,\textrm{a}}}(1 - \delta_{m}^2)| \hat{g}^{\textrm{V}}_{m,k} | ^2 \\
   & = {e}_{nm, k} + \frac{p^{\textrm{V}}_{m,\textrm{a}} L^{\textrm{V}}_{m, \textrm{a}} }{p^{\textrm{I}}_{n,\textrm{a}} L^{\textrm{I}}_{nm,\textrm{a}} }(1 - \delta_{m}^2) | {e}_{m, k} |^2.
\end{aligned}
\end{equation}

In \eqref{data sample}, ${e}_{nm, k}$ and ${e}_{m, k}$ are the realizations of the error in \eqref{i v2v} and \eqref{i v2i} at time slot $k$. Note that neither ${e}_{nm, k}$ nor ${e}_{m, k}$ can be obtained by the \ac{RSU}; however, the value of $z_k$ is accessible since all parameters in the first equation of \eqref{data sample} are known. Due to the \ac{i.i.d.} error ${e}_{nm, k}$ and ${e}_{m, k}$, $\mathcal{Z}$ is essentially a sequence of \ac{i.i.d.} samples from \ac{RV} $Z = {e}_{nm} + Y$ with ${e}_{nm} \sim \mathcal{E}_{nm}$, $Y = \lambda_Y^{-1} | {e}_{m} |^2 \sim \text{exp} (\lambda_Y)$, and $\lambda_Y = \frac{p^{\textrm{I}}_{n,\textrm{a}} L^{\textrm{I}}_{nm,\textrm{a}} }{p^{\textrm{V}}_{m,\textrm{a}} L^{\textrm{V}}_{m,\textrm{a}} (1 - \delta_{m}^2)}$.
Since $Z$ is the sum of two independent \ac{RV} ${e}_{nm}$ and $Y$, deriving the \ac{PDF} of $\mathcal{E}_{nm}$ through $\mathcal{Z}_m$ is essentially a deconvolution problem \cite{10.3150/08-BEJ146}. We define $f_E(e_{nm})$ as the \ac{PDF} of $\mathcal{E}_{nm}$, which is simplified to $f_{E,m}$ hereinafter for clarity. Given the Fourier transforms of $f_Z$, $f_{E,m}$, and $f_Y$ denoted by $F\left\{ f_Z \right\}$, $F\left\{ f_{E,m} \right\}$, $F\left\{ f_Y \right\}$, we can apply the Parseval's theorem and approximate $F\left\{ f_{E,m} \right\}$ as follows:
\begin{equation}
\begin{aligned}
\label{Fourier}
    F\left\{ f_{E,m} \right\} 
    = & \frac{F\left\{ f_Z \right\}}{F\left\{ f_Y \right\}} \approx \frac{1}{T}\sum_{k=1}^T e^{-jwz_{m,k}}(1+\frac{jw}{\lambda_Y}),
\end{aligned}
\end{equation}
where $F\left\{ f_Z \right\}$ is approximated by its empirical counterpart $\frac{1}{T}\sum_{k=1}^T e^{-jwz_k}$ from $\mathcal{Z}_m$ \cite{doi:10.1080/02331889008802238}. From \eqref{Fourier}, we can obtain $\hat{f}_{E,m}$, the estimation of $f_{E,m}$, by the inverse Fourier transform:
\begin{equation}
\begin{aligned}
\label{hat pdf}
    \hat{f}_{E,m} 
    = & \frac{1}{2 \pi T}  \sum_{k=1}^T \int_{-\infty}^{\infty}  e^{-jw(z_{m,k}-e_{nm})}  (1+\frac{jw}{\lambda_Y}) dw \\
    \approx & \frac{1}{2 \pi T}  \sum_{k=1}^T \int_{-K\pi}^{K \pi}  e^{-jw(z_{m,k}-e_{nm})}  (1+\frac{jw}{\lambda_Y}) dw,
\end{aligned}
\end{equation}
where the integral is truncated with constant $K$ to ensure the convergence of \eqref{hat pdf} \cite{10.3150/08-BEJ146}.
\begin{figure*}[t]
\vspace{-10pt}
	\centering
	\includegraphics[width=0.68\textwidth]{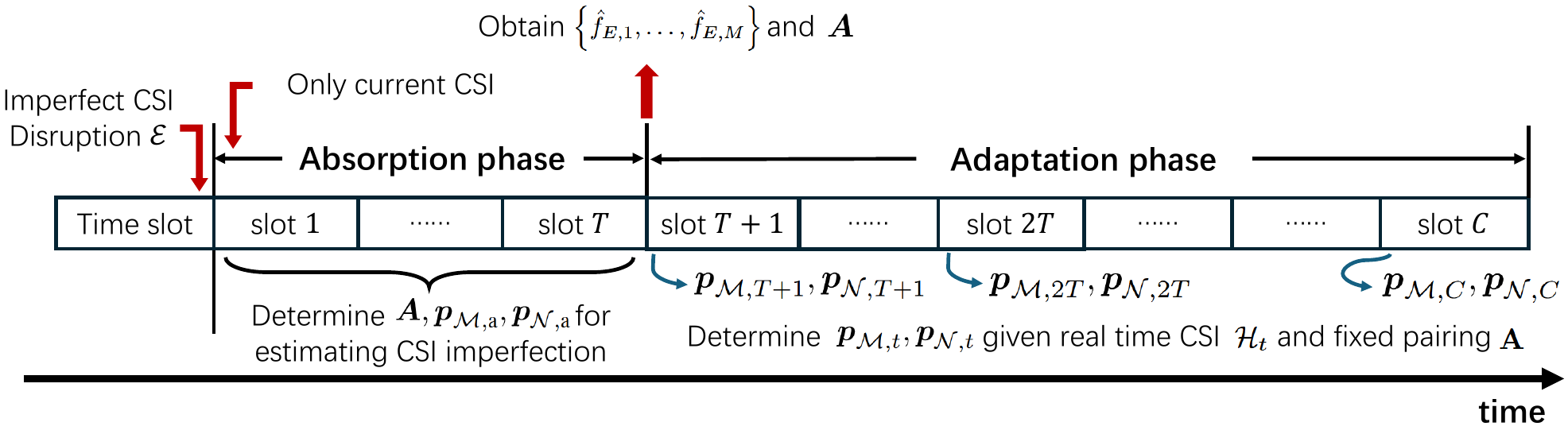}
    \vspace{-10pt}
	\caption{\small{The proposed two-phase resilient design for solving the bi-level problem \eqref{opt1}.}}
 \label{timescale_two_phase}
 \vspace{-15pt}
\end{figure*}

\vspace{-8pt}
\subsection{Adaptation Capability}
Given the estimator $\hat{f}_{E,m}$ in \eqref{hat pdf}, the \ac{RSU} can solve \eqref{c1-2} after time slot $t = T$, for a given matching $\boldsymbol{A} $. However, $\boldsymbol{A} $ should be determined at $t = 1$ for optimizing the upper-level problem \eqref{opt1} over an extended period of $C$ time slots. This is challenging since $J_t(\boldsymbol{p}^{*}_{\mathcal{M},t}, \boldsymbol{p}^{*}_{\mathcal{N},t}, \boldsymbol{A} ), \forall t > 1$ in \eqref{obj1} is unknown at $t = 1$. We observe that $J_t(\boldsymbol{p}_{\mathcal{M},t}, \boldsymbol{p}_{\mathcal{N},t}, \boldsymbol{A} )$ is fundamentally determined by the accuracy of estimated \ac{PDF}. Thus, a more accurate estimation on $\left\{\hat{f}_{E,1}, \ldots, \hat{f}_{E,M}\right\}$ yields lower $\sum_{t=1}^{C} J_t(\boldsymbol{p}_{\mathcal{M},t}, \boldsymbol{p}_{\mathcal{N},t}, \boldsymbol{A} )$. Hence, we define the overall \ac{MSE} of the \ac{PDF} estimators, i.e., $ \sum_{m \in \mathcal{M}} \mathbb{E}\left[  (f_{E,m} - \hat{f}_{E,m})^2 \right]$, as the \emph{adaptation capability} of the \ac{C-V2X} network, whose upper bound is derived.
\begin{theorem}
\label{theorem 1}
An upper bound on the \ac{MSE} of $\hat{f}_{E,m}$ is given by:
\begin{equation}
\small
\label{bound}
    \begin{aligned}
        &\mathbb{E}\left[  (f_{E,m} - \hat{f}_{E,m})^2 \right]
        \leq \frac{1}{4\pi^2} \left( \int_{w\geq |K\pi|} e^{jwe_{nm}} F\left\{ f_{E,m} \right\} dw \right)^2 \\
        & + \frac{K^2}{4T}\left[ \sqrt{1+ \beta_m^2 o_{nm}^2} + \frac{\ln{\left(\beta_m o_{nm} + \sqrt{1+ \beta_m^2 o_{nm}^2}\right)}}{\beta_m o_{nm}} \right]^2,
    \end{aligned}
\end{equation}
where $ \beta_m = K \pi (1-\delta_m^2)$ and $o_{nm} = \frac{p^{\textrm{V}}_{m,\textrm{a}} L^{\textrm{V}}_{m, \textrm{a}}}{p^{\textrm{I}}_{n,\textrm{a}} L^{\textrm{I}}_{nm,\textrm{a}}}$.
\begin{proof}
See the proof in the conference version \cite{shui2024resilienceperspectivecv2xcommunication}.
\end{proof}
\end{theorem}

Theorem \ref{theorem 1} shows that the adaptation capability of the \ac{C-V2X} network is upper bounded by the sum of two terms: the first one related to the unknown error distribution $\mathcal{E}_{nm}$ and the second one as a function of the absorption power scheme $p^{\textrm{V}}_{m,\textrm{a}}$ and $p^{\textrm{I}}_{n,\textrm{a}}$. Actually, Theorem \ref{theorem 1} provides an alternative objective function to the upper-level problem in \eqref{opt1}, without need for future \ac{CSI}. Specifically, we can set a high $K$ such that the first term in \eqref{bound} is negligible. Then, minimizing the original objective function $\sum_{t=1}^{C} J_t(\boldsymbol{p}^{*}_{\mathcal{M},t}, \boldsymbol{p}^{*}_{\mathcal{N},t}, \boldsymbol{A} )$ is equivalent to minimizing the second term in \eqref{bound}, where only the large-scale fading parameters $L^{\textrm{V}}_{m,\textrm{a}}$ and $L^{\textrm{I}}_{nm,\textrm{a}}$ during absorption are needed. From \eqref{bound}, we can further observe that the second term is monotonously increasing with $o_{nm}$. Thus, a high power $p^{\textrm{I}}_{n,\textrm{a}}$ on the \ac{V2I} link can enhance the \ac{C-V2X}'s adaptation capability by obtaining an accurate estimation on $f_{E,m}$. This is because the error $e_{nm}$ becomes the dominant component in $Z$ when a high power $p^{\textrm{I}}_{n,\textrm{a}}$ is applied. Conversely, employing a high power $p^{\textrm{V}}_{m,\textrm{a}}$ on the \ac{V2V} link will compromise the \ac{C-V2X}'s adaptation capability, since $Y$ becomes the dominant component in $Z$ other than $e_{nm}$, which, in turn, decreases the accuracy of the estimation $\hat{f}_{E,m}$. Another insight about enhancing the \ac{C-V2X}'s adaptation capability is that the matching of \ac{V2I} link and \ac{V2V} links $\boldsymbol{A} $ should be carefully designed, since the parameters $L^{\textrm{V}}_{m,\textrm{a}}$, $L^{\textrm{I}}_{nm,\textrm{a}}$, and $\delta_m$ will affect the adaptation capability. Moreover, we can see that, as the absorption phase lasts longer, i.e., a higher $T$ is allowed, the system's adaptation capability can be improved. 

As shown in Theorem \ref{theorem 1}, the design of the absorption power scheme $p^{\textrm{V}}_{m,\textrm{a}}$ and $p^{\textrm{I}}_{n,\textrm{a}}$ reveals a tradeoff between the \ac{C-V2X}'s adaptation capability and the \ac{QoS} of \ac{V2V} links during absorption. One may simply implement the minimal $p^{\textrm{V}}_{m,\textrm{a}}$ and the maximal $p^{\textrm{I}}_{n,\textrm{a}}$ to obtain an accurate estimation $\hat{f}_{E,m}$, which, however, will jeopardize the delay on the \ac{V2V} link $m$. Existing studies in \cite{9382930, 9857930, 10238756, 10213228} fail to address this tradeoff as they primarily focus on a desired \ac{QoS}, i.e., a high adaptation capability, while ignoring the \ac{QoS} during absorption. In a resilient design, the \ac{C-V2X} network is expected to achieve a high adaptation capability, on the condition that the \ac{QoS} during absorption is not significantly compromised. Thus, we need a new metric to capture the \ac{QoS} during absorption and demonstrate the interplay between absorption and adaptation.
\vspace{-8pt}
\subsection{\Ac{HR} during absorption}

From a resilience perspective, we adopt the concept of \emph{hazard rate} \cite{Brody2007} to evaluate the system's \ac{QoS} during the absorption phase. The key idea is that, while some short-term \ac{QoS} degradation during absorption may be acceptable if it enables high adaptation capability on the long run, it is equally important to limit this degradation to prevent severe consequences. As a simple example, if the \ac{V2V} delay requirement is $\tau_0 = 10$~ms, a system that experiences delay within $[10,20]$~ms is preferable to one where the delay fluctuates in the range of $[30,40]$~ms.
Formally, given the \ac{QoS} requirement $\tau_0$ on the \ac{V2V} links, the \ac{HR} is defined as
\begin{equation}
\label{hr v2v}
    \Lambda(\tau_0) \triangleq \lim_{\Delta \tau \rightarrow 0} \frac{\mathbb{P} \left\{ \tau_0 \leq \tau \leq \tau_0 + \Delta \tau \right\}}{\Delta \tau \mathbb{P} \left\{ \tau \geq \tau_0 \right\}},
\end{equation}
where the probability is taken with respect to the the small-scale fading during absorption.
By rewriting the definition in \eqref{hr v2v} as $\Lambda(\tau_0) = \lim_{\Delta \tau \rightarrow 0} \frac{\mathbb{P} \left\{ \tau_0 \leq \tau \leq \tau_0 + \Delta \tau | \tau_0 \leq \tau \right\}}{\Delta \tau }$, \ac{HR} actually quantifies the system’s capability to sustain \ac{QoS} close to the specified \ac{QoS} requirements, given that the \ac{QoS} requirement has already been violated. 
In other words, conditional on the \ac{QoS} requirement not being satisfied, a high \ac{HR} during absorption  ensures a high probability of maintaining the \ac{QoS} near the specified requirement. Essentially, a high \ac{HR} $\Lambda(\tau_0)$ implies an increased likelihood of preserving $\tau$ close to $\tau_0$ when $\tau > \tau_0$. The explicit expression of \eqref{hr v2v} is derived next.
\begin{lemma}
\label{lemma 1}
If \ac{V2V} link $m$ is reusing the \ac{RB} of \ac{V2I} link $n$, i.e., $\alpha_{mn} = 1$, the \ac{HR} $\Lambda_m$ on \ac{V2V} link $m$ is given by:
\begin{equation}
\vspace{-5pt}
\label{HR v2v}
\begin{aligned}
    \Lambda_{m} = D_{\textrm{V}} e^{-\frac{\sigma^2\gamma_{\textrm{V}}}{p^{\textrm{V}}_{m,\textrm{a}}L^{\textrm{V}}_{m,\textrm{a}}}} \frac{ \frac{p^{\textrm{I}}_{n,\textrm{a}}L^{\textrm{I}}_{nm,\textrm{a}}}{p^{\textrm{V}}_{m,\textrm{a}}L^{\textrm{V}}_{m,\textrm{a}}} + \frac{\sigma^2}{p^{\textrm{V}}_{m,\textrm{a}} L^{\textrm{V}}_{m,\textrm{a}}} \left( 1 + \frac{p^{\textrm{I}}_{n,\textrm{a}}L^{\textrm{I}}_{nm,\textrm{a}}}{p^{\textrm{V}}_{m,\textrm{a}}L^{\textrm{V}}_{m,\textrm{a}}} \gamma_{\textrm{V}} \right)  }{\left( 1 + \frac{p^{\textrm{I}}_{n,\textrm{a}}L^{\textrm{I}}_{nm,\textrm{a}}}{p^{\textrm{V}}_{m,\textrm{a}}L^{\textrm{V}}_{m,\textrm{a}}} \gamma_{\textrm{V}} - e^{-\frac{\sigma^2\gamma_{\textrm{V}}}{p^{\textrm{V}}_{m,\textrm{a}}L^{\textrm{V}}_{m,\textrm{a}}}} \right)^2 },
\end{aligned}
\end{equation}
where $D_{\textrm{V}} = \frac{\ln{2}D2^{\frac{D}{B\tau_0}}}{B\tau_0^2} $.
\end{lemma}
\begin{proof}
See Appendix \ref{Proof of Lemma 1}.
\end{proof}
Aligned with resilience, \ac{HR} is not intended to strictly prevent \ac{QoS} degradation, but rather to mitigate the degradation's severity. Thus, Lemma \ref{lemma 1} provides a guidance on the design of the absorption power scheme $p^{\textrm{V}}_{m,\textrm{a}}$ and $p^{\textrm{I}}_{n,\textrm{a}}$. Precisely, a high \ac{HR} can keep an acceptable \ac{QoS} degradation during absorption while enhancing the \ac{C-V2X} network's adaptation capability.

\vspace{-10pt}
\subsection{matching and Power Optimization during absorption}
\begin{figure}[t]
	\centering
	\includegraphics[width=0.32\textwidth]{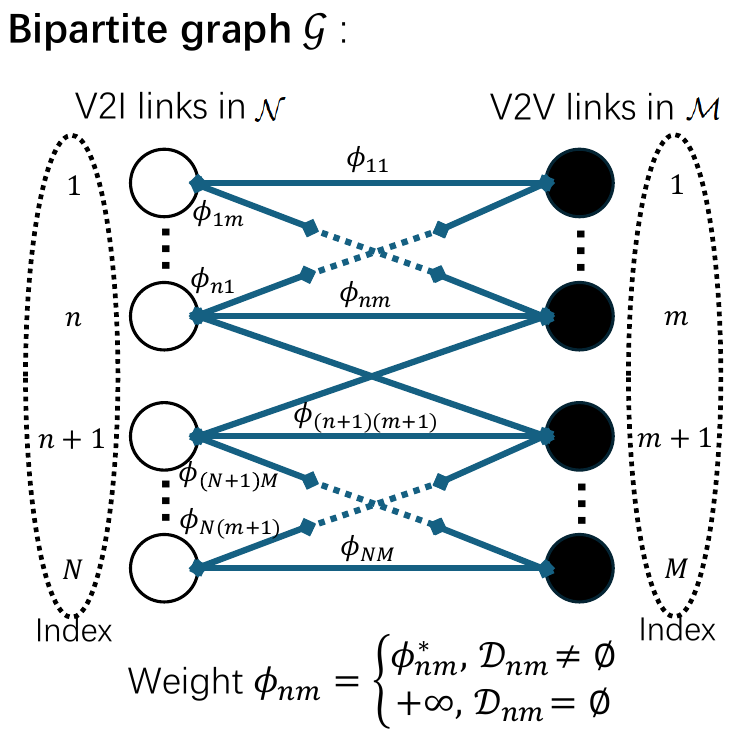}
    \vspace{-8pt}
	\caption{\small{The bipartite graph $\mathcal{G} = (\mathcal{M} \times \mathcal{N}, \mathcal{E}_{\mathcal{G}})$ used for solving \eqref{opt2}.}}
 \label{Hungarian}
     \vspace{-10pt}
\end{figure}
Given the results in \eqref{bound} and \eqref{HR v2v}, we can reformulate the upper-level problem of \eqref{opt1}:
\begin{subequations}
\label{opt2}
\begin{IEEEeqnarray}{s,rCl'rCl'rCl}
& \underset{\boldsymbol{p}_{\mathcal{M},\textrm{a}}, \boldsymbol{p}_{\mathcal{N}, \textrm{a}}, \boldsymbol{A} }{\text{min}} &\quad& \sum_{m \in \mathcal{M}} \mathbb{E}\left[  (f_{E,m} - \hat{f}_{E,m})^2 \right] \label{obj2}\\
&\text{s.t.} && \eqref{c1-1}, \nonumber \\
&&& \Lambda_{m} \geq \lambda_m \Lambda_{m,\textrm{max}} , \forall m \in \mathcal{M} \label{c2-2}, \\
&&& p^{\textrm{V}}_{\textrm{min}} \leq p^{\textrm{V}}_{m,\text{a}} \leq p^{\textrm{V}}_{\textrm{max}}, \forall m \in \mathcal{M} \label{c2-3}, \\
&&& p^{\textrm{I}}_{\textrm{min}} \leq p^{\textrm{I}}_{n,\text{a}} \leq p^{\textrm{I}}_{\textrm{max}}, \forall n \in \mathcal{N} \label{c2-4}, 
\end{IEEEeqnarray}
\end{subequations}
where $\boldsymbol{p}_{\mathcal{M},\textrm{a}} $, $\boldsymbol{p}_{\mathcal{N}, \textrm{a}}$, and matching $\boldsymbol{A} $ are optimized during the absorption phase. 
In \eqref{c2-2}, $\boldsymbol{\lambda} = \left[ \lambda_1, \ldots, \lambda_M \right]$ are predefined by the \ac{RSU} to balance the \ac{C-V2X} network's prioritization on its adaptation capability and \ac{QoS} during absorption. A higher $\lambda_m$ ensures that the \ac{QoS} of \ac{V2V} link $m$ is less compromised during absorption, however, at the expense of a lower accuracy of estimation $\hat{f}_{E,m}$. To solve \eqref{opt2}, we approximate $\Lambda_{m} \approx \frac{o_{nm} D_{\textrm{V}}}{\gamma_{\textrm{V}}^2 }$ by ignoring the noise term $\sigma^2$ in \eqref{HR v2v}. This approximation is helpful in addressing the non-convexity of \eqref{HR v2v}. Thus, we have $\Lambda_{m,\textrm{max}} \approx \frac{D_{\textrm{V}} p^{\textrm{V}}_{\textrm{max}} L^{\textrm{V}}_{m,\textrm{a}} }{\gamma_{\textrm{V}}^2 p^{\textrm{I}}_{\textrm{min}} L^{\textrm{I}}_{nm,\textrm{a}}}$ under constraints \eqref{c2-3} and \eqref{c2-4}. 
\begin{figure}[t]
	\centering
	\includegraphics[width=0.40\textwidth]{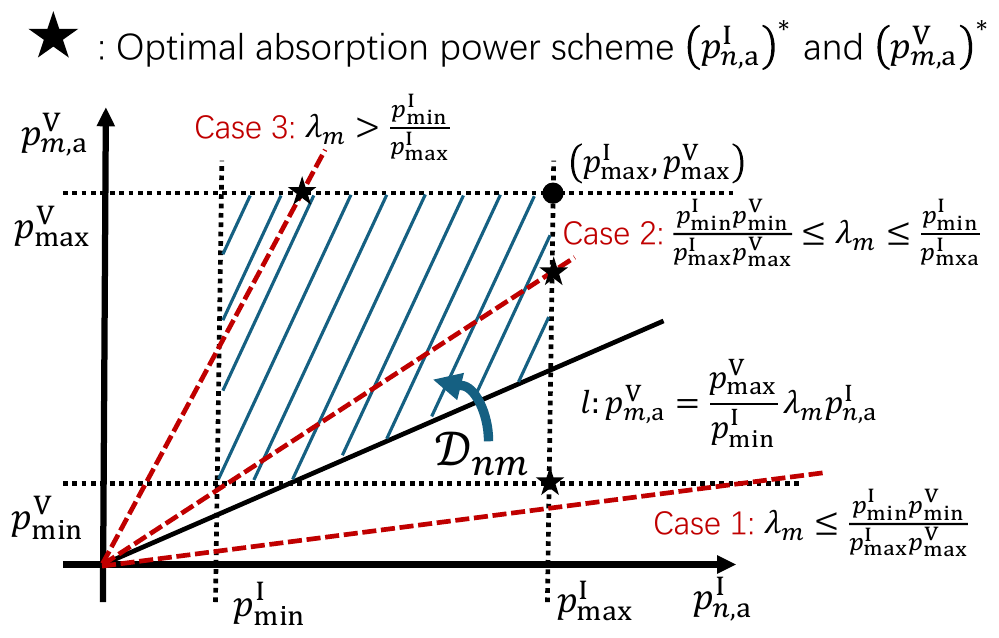}
    \vspace{-8pt}
	\caption{\small{The different cases of optimal absorption power scheme.}}
 \label{optimal power a}
      \vspace{-8pt}
\end{figure}
Now, \eqref{opt2} can be transformed into a minimum weight matching problem. Specifically, we can construct a bipartite graph $\mathcal{G} = (\mathcal{M} \times \mathcal{N}, \mathcal{E}_{\mathcal{G}})$ where $\mathcal{E}_{\mathcal{G}}$ is the set of edges that connect to the vertices (vehicular links) from set $\mathcal{M}$ and $\mathcal{N}$, as shown in Fig. \ref{Hungarian}. The weight of the edge that connects $n \in \mathcal{N}$ and $m \in \mathcal{M}$ is given by:
\begin{equation}
\label{weight}
\phi_{nm} =
\begin{cases}
    \phi^{*}_{nm}, & \mathcal{D}_{nm} \neq \varnothing, \\
    + \infty, & \text{otherwise,}
\end{cases}
\end{equation}
where $\phi^{*}_{nm}$ and $\mathcal{D}_{nm}$ are given in \eqref{phi} and \eqref{phi feasible}.
\begin{figure*}
\vspace{-10pt}
    \begin{equation}
\label{phi}
    \phi^{*}_{nm} = \underset{ (p^{\textrm{V}}_{m,\text{a}}, p^{\textrm{I}}_{n,\text{a}}) \in \mathcal{D}_{nm}}{\text{min}} \left[ \sqrt{1+ \beta_m^2 o_{nm}^2} + \frac{\ln{\left(\beta_m o_{nm} + \sqrt{1+ \beta_m^2 o_{nm}^2}\right)}}{\beta_m o_{nm}} \right]^2,
\end{equation}
\begin{equation}
\label{phi feasible}
\mathcal{D}_{nm} = 
\left\{ (p^{\textrm{V}}_{m,\text{a}}, p^{\textrm{I}}_{n,\text{a}}) \middle|
\begin{array}{l}
\frac{p^{\textrm{V}}_{\textrm{max}}}{ p^{\textrm{I}}_{\textrm{min}}}  \lambda_m \leq \frac{p^{\textrm{V}}_{m,\textrm{a}}}{p^{\textrm{I}}_{n,\textrm{a}} } , \
p^{\textrm{V}}_{\textrm{min}} \leq p^{\textrm{V}}_{m,\text{a}} \leq p^{\textrm{V}}_{\textrm{max}}, \ \text{and} \
p^{\textrm{I}}_{\textrm{min}} \leq p^{\textrm{I}}_{n,\text{a}} \leq p^{\textrm{I}}_{\textrm{max}}
\end{array} 
\right\}. 
\end{equation}
\vspace{-20pt}
\end{figure*}
In \eqref{phi}, the original function $\mathbb{E}\left[  (f_{E,m} - \hat{f}_{E,m})^2 \right]$ is replaced by the second term of its upper bound in Theorem 1 and \eqref{phi feasible} is directly obtained from \eqref{c2-2}, \eqref{c2-3}, and \eqref{c2-4}. Since the objective function in \eqref{phi} is monotonously increasing with $o_{nm} = \frac{p^{\textrm{V}}_{m,\textrm{a}} L^{\textrm{V}}_{m,\textrm{a}}}{p^{\textrm{I}}_{n,\textrm{a}} L^{\textrm{I}}_{nm,\textrm{a}}}$, we can first find $o_{nm}$ that minimizes \eqref{phi} in $\mathcal{D}_{nm}$. Then, we will obtain multiple optimal absorption transmit power scheme $(p^{\textrm{V}}_{m,\text{a}})^{*}$ and $(p^{\textrm{I}}_{n,\text{a}})^{*}$ since $o_{nm}$ is only determined by the ratio of $p^{\textrm{I}}_{n,\text{a}}$ and $p^{\textrm{V}}_{m,\text{a}}$. To back up our approximation of \ac{HR} by ignoring noise, we choose the optimal solution with the highest transmit power, which is given by:
\begin{equation}
\label{power scheme}
\left( (p^{\textrm{I}}_{n,\text{a}})^{*}, (p^{\textrm{V}}_{m,\text{a}})^{*}\right)  =
\begin{cases}
    \left(p^{\textrm{I}}_{\textrm{max}}, p^{\textrm{V}}_{\textrm{min}} \right), 
    &  \lambda_m \leq \frac{p^{\textrm{I}}_{\textrm{min}} p^{\textrm{V}}_{\textrm{min}}}{p^{\textrm{I}}_{\textrm{max}} p^{\textrm{V}}_{\textrm{max}}}, \\
    \left(p^{\textrm{I}}_{\textrm{max}}, \frac{p^{\textrm{I}}_{\textrm{max}} p^{\textrm{V}}_{\textrm{max}} \lambda_m }{p^{\textrm{I}}_{\textrm{min}}} \right),
    & \frac{p^{\textrm{I}}_{\textrm{min}} p^{\textrm{V}}_{\textrm{min}}}{p^{\textrm{I}}_{\textrm{max}} p^{\textrm{V}}_{\textrm{max}}} < \lambda_m \leq \frac{p^{\textrm{I}}_{\textrm{min}}}{p^{\textrm{I}}_{\textrm{max}}}, \\
    \left(\frac{p^{\textrm{I}}_{\textrm{max}}}{\lambda_m} , p^{\textrm{V}}_{\textrm{max}} \right),
    & \frac{p^{\textrm{I}}_{\textrm{min}}}{p^{\textrm{I}}_{\textrm{max}}} < \lambda_m.
\end{cases}
\end{equation}
The different cases in \eqref{power scheme} are illustrated in Fig. \ref{optimal power a} with the complete algorithm for solving $\eqref{opt2}$ given in Algorithm \ref{algo1}. Although the closed-form solutions provided in \eqref{power scheme} are only sub-optimal with respect to the original upper-level problem in \eqref{opt1}, they represent the best attainable solutions given that the matching $\boldsymbol{A}$ must be determined at time slot $t = 1$. Moreover, the \ac{RSU} can exploit \eqref{power scheme} to accurately estimate the \ac{PDF}, thereby reducing the original objective function in \eqref{opt1}.

\vspace{-10pt}
\section{Adaptation Phase}
\label{adaptation section}
After obtaining the matching $\boldsymbol{A} $ and the estimated \ac{PDF} in the absorption phase, the \ac{C-V2X} network switches to the adaptation phase. During adaptation, the transmit power schemes $\boldsymbol{p}_{\mathcal{M},t}$ and $\boldsymbol{p}_{\mathcal{N},t}$ are optimized at time slot $t$ for satisfying the \ac{QoS} requirements according to real-time imperfect \ac{CSI} $\hat{\mathcal{H}}_t$. Particularly, with $\boldsymbol{A} $ given in \eqref{opt2}, the lower-level problem in \eqref{c1-2} can be transformed into $M$ individual sub-problems, where the $m$-th sub-problem focuses on \ac{V2V} link $m$ and its matched \ac{V2I} link. We consider a general case where \ac{V2V} link $m$ is reusing the \ac{RB} of \ac{V2I} link $n$. Then, the $m$-th sub-problem at time slot $t$ is given by:
\begin{subequations}
\label{opt4}
\begin{IEEEeqnarray}{s,rCl'rCl'rCl}
& \underset{p^{\textrm{V}}_{m,t}, p^{\textrm{I}}_{n,t}}{\text{min}} &\quad& J_{m,t}(p^{\textrm{V}}_{m,t}, p^{\textrm{I}}_{n,t})  \label{obj4}\\
&\text{s.t.} && R_{n,t}( p^{\textrm{V}}_{m,t}, p^{\textrm{I}}_{n,t}) \geq R_0 \label{c4-1}, \\
&&& P_{m,t}(\hat{\mathcal{E}}_{nm}, p^{\textrm{V}}_{m,t}, p^{\textrm{I}}_{n,t}) \geq P_0 \label{c4-2}, \\
&&& p^{\textrm{V}}_{\textrm{min}} \leq p^{\textrm{V}}_{m,t} \leq p^{\textrm{V}}_{\textrm{max}} \ \text{and} \ p^{\textrm{I}}_{\textrm{min}} \leq p^{\textrm{I}}_{n,t} \leq p^{\textrm{I}}_{\textrm{max}} \label{c4-3},
\end{IEEEeqnarray}
\end{subequations}
where $ R_{n,t}( p^{\textrm{V}}_{m,t}, p^{\textrm{I}}_{n,t})$ and $P_{m,t}(\hat{\mathcal{E}}_{nm}, p^{\textrm{V}}_{m,t}, p^{\textrm{I}}_{n,t})$ are the throughput on \ac{V2I} link $n$ and the probability of satisfying the delay on \ac{V2V} link $m$ in the considered case. For notational simplicity, we replace $P_{m,t}(\hat{\mathcal{E}}_{nm}, p^{\textrm{V}}_{m,t}, p^{\textrm{I}}_{n,t}) $ by $\hat{P}_{m,t}^{(nm)}$ and $P_{m,t}(\mathcal{E}_{nm}, p^{\textrm{V}}_{m,t}, p^{\textrm{I}}_{n,t})$ by $P_{m,t}^{(nm)}$ in the following analysis. Then, the objective function \eqref{obj4} is given as $
J_{m,t}(p_{m,t}, p_{n,t}) = \left[\hat{P}_{m,t}^{(nm)} - P_{m,t}^{(nm)} \right]^2$.

To solve \eqref{opt4}, we begin by deriving the expression of \eqref{obj4}:
\begin{equation}
\label{prob analytic}
\begin{aligned}
    P_{m,t}^{(nm)} & = \mathbb{P} \Biggl\{\frac{p^{\textrm{V}}_{m,t} L^{\textrm{V}}_{m,t}}{ \gamma_{\textrm{V}} } (1 - \delta_{m,t}^2) | {e}_{m,t} |^2  \\
    & \quad -  p^{\textrm{I}}_{n,t} L^{\textrm{I}}_{nm,t} {e}_{nm,t} \geq b_{m,t} \Bigm| {e}_{nm,t} \sim \mathcal{E}_{nm}\Biggr\} \\
    & \overset{(a)}{=}  1 - \int_{-\frac{b_t}{p^{\textrm{I}}_{n,t} L^{\textrm{I}}_{nm,t}}}^{\infty}  \Biggl[1 \\ 
    & \quad - \exp{\left( -\frac{\gamma_{\textrm{V}} (b_{m,t}+ p^{\textrm{I}}_{n,t} L^{\textrm{I}}_{nm,t} x)}{p^{\textrm{V}}_{m,t} L^{\textrm{V}}_{m,t} (1-\delta_{m,t}^2) }   \right) }\Biggr] f_{E,m}(x) dx \\
    & \overset{(b)}{\approx} 1 - \int_{-\frac{b_t}{p^{\textrm{I}}_{n,t} L^{\textrm{I}}_{nm,t}}}^{-\frac{b_t}{p^{\textrm{I}}_{n,t} L^{\textrm{I}}_{nm,t}}+ K_1 } \Biggl[1 \\
    & \quad - \exp{\left( -\frac{\gamma_{\textrm{V}} (b_{m,t}+ p^{\textrm{I}}_{n,t} L^{\textrm{I}}_{nm,t} x)}{p^{\textrm{V}}_{m,t} L^{\textrm{V}}_{m,t} (1-\delta_{m,t}^2) }   \right) }\Biggr] f_{E,m}(x) dx \\
    & \overset{(c)}{=} 1 - \int_{-\infty}^{\infty}  \Phi_t(x) f_{E,m}(x) dx \\
    & \overset{(d)}{=} 1 -  \frac{1}{2\pi} \int_{-\infty}^{\infty} \left(F\left\{ \Phi_t \right\} \frac{F^{*}\left\{ f_Z \right\}}{F^{*}\left\{ f_Y \right\}} \right)(w) dw.
\end{aligned}
\end{equation}

In \eqref{prob analytic}, (a) is based on the convolution rule for the \ac{PDF} when summing two independent \ac{RV}. The approximation (b) is adopted for analytical tractability in the following discussion\footnote{This approximation is valid when $K_1$ is large enough since $\underset{x \rightarrow \infty}{\lim} \exp{\left( -\frac{\gamma_{\textrm{V}} (b_{m,t}+ p^{\textrm{I}}_{n,t} L^{\textrm{I}}_{nm,t} x)}{p^{\textrm{V}}_{m,t} L^{\textrm{V}}_{m,t} (1-\delta_{m,t}^2) }   \right) } = 0$ and $\underset{x \rightarrow \infty}{\lim}f_{E,m}(x) = 0$.}. In (c), we define $\Phi_t(x) = \left[1 - \exp{\left( -\frac{\gamma_{\textrm{V}} \left(b_t+ p^{\textrm{I}}_{n,t} L^{\textrm{I}}_{nm,t} x\right)}{p^{\textrm{V}}_{m, t} L^{\textrm{V}}_{m, t} \left(1-\delta_{m,t}^2\right) } \right) } \right] I(x)$ with
\begin{equation}
I(x) = \begin{cases} 
1, & \text{if } -\frac{b_t}{p^{\textrm{I}}_{n,t} L^{\textrm{I}}_{nm,t}} \leq x \leq -\frac{b_{t}}{p^{\textrm{I}}_{n,t} L^{\textrm{I}}_{nm,t}} + K_1, \\
0,  & \text{otherwise }.
\end{cases}
\end{equation}
In (d), we leverage Parseval's theorem.
According to \eqref{prob analytic}, we can replace $F^{*}\left\{ f_Z \right\} $ with its empirical estimation $\frac{1}{T}\sum_{k=1}^T e^{-jwz_k}$ and obtain the expression of $\hat{P}_{m,t}^{(nm)}$ as
\begin{equation}
\label{prob hat}
\begin{aligned}
    \hat{P}_{m,t}^{(nm)} = 1 - \frac{1}{2\pi T}\sum_{k=1}^T \int_{-\infty}^{\infty} F\left\{ \Phi_t \right\} e^{jwz_k}(1-\frac{jw}{\lambda_Y}) dw, 
\end{aligned}
\end{equation}
where
\begin{equation}
\label{F}
\begin{aligned}
    F\left\{ \Phi_t \right\} 
    = & \frac{e^{-jwK_1} - 1}{-jw} e^{jw\frac{b_t}{p^{\textrm{I}}_{n,t} L^{\textrm{I}}_{nm,t}}} \\
    & + \frac{e^{-(c_t+jw)K_1 }- 1}{c_t+jw} e^{ jw\frac{b_t}{p^{\textrm{I}}_{n,t} L^{\textrm{I}}_{nm,t}}} \\
    \approx & \frac{e^{-jwK_1} - 1}{-jw} e^{jw\ell(c_t)} + \frac{e^{-(c_t+jw)K_1 }- 1}{c_t+jw} e^{ jw\ell(c_t)}.
\end{aligned}
\end{equation}
\begin{algorithm}[t]
\footnotesize
\caption{absorption power scheme and matching optimization in absorption phase}
\label{algo1}
\begin{algorithmic}[1] 
\Require Large-scale fading $L^{\textrm{I}}_{nm, \textrm{a}}$, $L^{\textrm{V}}_{m, \textrm{a}}$ and \ac{V2V} link coefficients $\delta_m$, $\forall m \in \mathcal{M}, \forall n \in \mathcal{N}$ and parameters $\boldsymbol{\lambda}$.
\Ensure absorption transmit power scheme $\boldsymbol{p}_{\mathcal{M},\textrm{a}}, \boldsymbol{p}_{\mathcal{N},\textrm{a}}$ and matching $\boldsymbol{A} $ 
\State Find optimal $ (p^{\textrm{I}}_{n,\text{a}})^{*}$, $ (p^{\textrm{V}}_{m,\text{a}})^{*}$ and $o_{nm}^{*}, \forall m \in \mathcal{M}, \forall n \in \mathcal{N}$ based on \eqref{power scheme} \;
\State Assign the value of $\phi_{nm}$ using $ o_{nm}^{*}, \forall m \in \mathcal{M}, \forall n \in \mathcal{N}$ based on \eqref{weight}\;
\State Find $\boldsymbol{A} $, $\boldsymbol{p}_{\mathcal{M},\textrm{a}}$ and $\boldsymbol{p}_{\mathcal{N},\textrm{a}}$ given the bipartite graph $\mathcal{G}$ by Hungarian method\; 
\end{algorithmic}
\end{algorithm}
The approximation in \eqref{F} relies on the assumption that the additive noise $\sigma^2$ in $b_t$ is negligible. In \eqref{F}, we define $c_t \triangleq \frac{ \gamma_{\textrm{V}} p^{\textrm{I}}_{n,t} L^{\textrm{I}}_{nm,t}}{p^{\textrm{V}}_{m,t} L^{\textrm{V}}_{m,t} (1-\delta_m^2)}$ and $\ell(c_t) \triangleq | \hat{g}^{\textrm{I}}_{nm,t} | ^2 - \frac{ | \hat{g}^{\textrm{V}}_{m,t} | ^2}{c_t} \frac{\delta_m^2}{1 - \delta_m^2} $. However, the integral in \eqref{prob hat} may not converge because of the term $1-\frac{jw}{\lambda_Y}$. To address this issue, we leverage the truncated regularization \cite{10.1214/08-AOS652} for the following approximation
\begin{equation}
\begin{aligned}
    \label{prob estimated}
     \hat{P}_{m,t}^{(nm)}
     \approx 1 - \frac{1}{2\pi T}\sum_{k=1}^T \int_{-K_2 \pi}^{K_2 \pi} F\left\{ \Phi_t \right\} e^{jwz_k}(1-\frac{jw}{\lambda_Y}) dw.
\end{aligned}
\end{equation}
Given \eqref{prob analytic} and \eqref{prob estimated}, the following theorem can be obtained
\begin{theorem}
The upper bound on the \ac{MSE} of $\hat{P}_{m,t}^{(nm)}$ is given by:
\begin{equation}
\label{qos upper}
    \begin{aligned}
         &\mathbb{E}\left[ \left( \hat{P}_{m,t}^{(nm)} -  P_{m,t}^{(nm)}\right)^2 \right] \leq \frac{1}{\pi^2T} \left[u(p^{\textrm{I}}_{n,t}, p^{\textrm{V}}_{m,t}) - 1\right]^2 \\
         & + \left( \frac{1}{2\pi} \int_{w\geq |K_2\pi|} F\left\{ \Phi_t \right\} F^{*}\left\{ f_{E,m} \right\} dw \right)^2,
    \end{aligned}
\end{equation}
where $u(p^{\textrm{I}}_{n,t}, p^{\textrm{V}}_{m,t})$ is given by:
\begin{equation}
\small
\label{u}
    \begin{aligned}
        u(p^{\textrm{I}}_{n,t}, p^{\textrm{V}}_{m,t}) = &\sqrt{1+\lambda_{Y}^{-2}K_2^2\pi^2} + \ln{\left(\frac{\sqrt{1+\lambda_{Y}^{-2}K_2^2\pi^2} - 1}{\lambda_{Y}^{-1}K_2\pi}\right)} \\
        & + c_t\lambda_{Y}^{-1} \ln{ \left(\frac{ K_2\pi + \sqrt{c_t^2 + K_2^2\pi^2}}{c_t}\right)} \\
        & + \frac{1}{c_t} \ln{\left( \frac{c_t K_2\pi}{\sqrt{K_2^2\pi^2 + c_t^2} + K_2\pi} \right)}. 
    \end{aligned}
\end{equation}
\begin{proof}
See Appendix \ref{Proof of Proposition 1}.
\end{proof}
\end{theorem}

Theorem 2 shows that the transmit powers $p^{\textrm{V}}_{m,t}$ and $p^{\textrm{I}}_{n,t}$ jointly determine the accuracy of the estimated probability $\hat{P}_{m,t}^{(nm)}$. Specifically, the second term on the right side of \eqref{qos upper} is related to the unknown distribution $\mathcal{E}_{nm}$ while the first term is determined by $p^{\textrm{V}}_{m,t}$ and $p^{\textrm{I}}_{n,t}$. At each time slot $t$ during adaptation, the optimal $p^{\textrm{V}}_{m,t}$ and $p^{\textrm{I}}_{n,t}$ should be optimized to minimize \eqref{qos upper}, which exactly corresponds to the lower-level problem in \eqref{c1-2}. We also observe that $u(p^{\textrm{I}}_{n,t}, p^{\textrm{V}}_{m,t})$ is not only determined by the current transmit powers $p^{\textrm{I}}_{n,t}$ and $ p^{\textrm{V}}_{m,t}$, but it is also affected by the absorption power scheme $p^{\textrm{V}}_{m,\text{a}}$ and $p^{\textrm{I}}_{n,\text{a}}$ according to $\lambda_Y = \frac{p^{\textrm{I}}_{n,\textrm{a}} L^{\textrm{I}}_{nm,\textrm{a}} }{p^{\textrm{V}}_{m,\textrm{a}} L^{\textrm{V}}_{m,\textrm{a}} (1 - \delta_{m}^2)}$. In other words, to effectively utilize \eqref{prob estimated} during adaptation, the \ac{C-V2X} network must account for the accuracy of the estimated error distribution during absorption. 
Thus, Theorem 2 provides critical guidance on the real-time transmit power design for recovering the \ac{C-V2X}'s \ac{QoS}.

According to the upper bound in Theorem 2, we can now show that the first integral term introduced by the truncation converges to $0$ asymptomatically:
\begin{corollary}
Given \eqref{F}, the following limit holds
\begin{equation}
    \lim_{K_2 \rightarrow \infty} \left( \frac{1}{2\pi} \int_{w\geq |K_2\pi|} F\left\{ \Phi_t \right\} F^{*}\left\{ f_{E,m} \right\} dw \right)^2 = 0.
\end{equation}
\end{corollary}
\noindent
Given Corollary 1, we can rewrite problem \eqref{opt4}:
\begin{subequations}
\label{opt5}
\begin{IEEEeqnarray}{s,rCl'rCl'rCl}
& \underset{p^{\textrm{V}}_{m,t}, p^{\textrm{I}}_{n,t}}{\text{min}} &\quad& \left[u(p^{\textrm{I}}_{n,t}, p^{\textrm{V}}_{m,t}) - 1\right]^2\label{obj5}\\
&\text{s.t.} && c_t \frac{(1-\delta_{m,t}^2) L^{\textrm{V}}_{m,t} L^{\textrm{I}}_{n,t} |g^{\textrm{I}}_{n,t}|^2}{\gamma_{\textrm{V}} L^{\textrm{I}}_{nm,t} L^{\textrm{V}}_{mn,t} |g^{\textrm{V}}_{mn,t}|^2} \geq 2^{\frac{R_0}{B}} - 1 \label{c5-1},  \\
&&& \eqref{c4-2} \text{ and } \eqref{c4-3} \nonumber, 
\end{IEEEeqnarray}
\end{subequations}
where \eqref{c5-1} is transformed from \eqref{c4-1} by ignoring the additive noise $\sigma^2$. To solve \eqref{opt5}, the monotonicity of $u(p^{\textrm{I}}_n, p^{\textrm{V}}_m)$ is shown next.
\begin{proposition}
Recall the definition $c_t =  \frac{ \gamma_{\textrm{V}} p^{\textrm{I}}_{n,t} L^{\textrm{I}}_{nm,t}}{p^{\textrm{V}}_{m,t} L^{\textrm{V}}_{m,t} (1-\delta_m^2)}$. Consequently, $u(p^{\textrm{I}}_{n,t}, p^{\textrm{V}}_{m,t})$ can be expressed as $u(c_t)$, where $u(c_t)$ is monotonically increasing with respect to $c_t$, if the truncation factor $K_2$ in \eqref{prob estimated} satisfies the following conditions:
\begin{equation}
\label{condition}
\begin{aligned}
    & \frac{x}{\sqrt{x^2 + D^{2}}} - \ln{\left(\frac{x + \sqrt{x^2 + D^{2}}}{D}\right)} \\
    & - \lambda_Yx^2 \ln{\left( x + \sqrt{x^2 + D^{2}} \right)} \leq 0, \forall x > 0, 
\end{aligned}
\end{equation}
where $D \triangleq (K_2\pi)^{-1}$ and $x \triangleq (c_t)^{-1}$. Moreover, the condition \eqref{condition} can be always satisfied in our settings since we can pick $K_2$ large enough, i.e., $D$ close to $0$.
\end{proposition}
\begin{proof}
See Appendix \ref{Proof of Theorem 2}.
\end{proof}
Given the result of Proposition 1, $u(c_t)$ can be viewed as an increasing function of $c_t$. However, problem \eqref{opt5} is still challenging since \eqref{c4-2} is complex due to \eqref{prob estimated}. Moreover, real-time solutions for \eqref{opt5} are needed due to the fast-changing nature of small-scale fading, where computational complexity should be reduced as much as possible. To address this issue, we approximate the monotonicity of \eqref{c4-2} in Corollary 2. Given Corollary 2, we reformulate problem \eqref{opt5}:
\begin{subequations}
\label{opt6}
\begin{IEEEeqnarray}{s,rCl'rCl'rCl}
& \underset{c_{t}}{\text{min}} &\quad& |u(c_t) - 1| \label{obj6}\\
&\text{s.t.} && c_t \frac{(1-\delta_m^2) L^{\textrm{V}}_{m,t} L^{\textrm{I}}_{n,t} |g^{\textrm{I}}_{n,t}|^2}{\gamma_{\textrm{V}} L^{\textrm{I}}_{nm,t} L^{\textrm{V}}_{mn,t} |g^{\textrm{V}}_{mn,t}|^2} \geq 2^{\frac{R_0}{B}} - 1 \label{c6-1},  \\
&&& \beta(c_t) \geq P_0 \label{c6-2}, \\
&&&  \frac{ \gamma_{\textrm{V}} p^{\textrm{I}}_{\textrm{min}} L^{\textrm{I}}_{nm,t}}{p^{\textrm{V}}_{\textrm{max}} L^{\textrm{V}}_{m,t} (1-\delta_{m,t}^2)}  \leq c_t \leq  \frac{ \gamma_{\textrm{V}} p^{\textrm{I}}_{\textrm{max}} L^{\textrm{I}}_{nm,t}}{p^{\textrm{V}}_{\textrm{min}} L^{\textrm{V}}_{m,t} (1-\delta_{m,t}^2)}  \label{c6-3},
\end{IEEEeqnarray}
\end{subequations}
where \eqref{c6-3} is derived from \eqref{c4-3} and the analytic expression of $\beta(c_t)$ is given in \eqref{beta_analytic} based on \eqref{F} and \eqref{prob estimated}.
\begin{figure*}[t]
\vspace{-10pt}
    \begin{equation}
\label{beta_analytic}
    \beta(c_t) = 1 - \frac{1}{2\pi T}\sum_{k=1}^T \int_{-K \pi}^{K \pi} \left[ \frac{e^{-jwK_1} - 1}{-jw} e^{jw\ell(c_t)} + \frac{e^{-(c_t+jw)K_1 }- 1}{c_t+jw} e^{ jw\ell(c_t)} \right] e^{jwz_k}(1-\frac{jw}{\lambda_Y}) dw.
\end{equation}
\vspace{-20pt}
\end{figure*}
\begin{corollary}
\label{corollary 2}
Provided the additive noise $\sigma^2$ is negligible, the probability $\hat{P}_{m,t}^{(nm)}$ in \eqref{prob estimated} can be approximated as follows:
\begin{equation}
\label{beta_c_t}
\begin{aligned}
\hat{P}_{m,t}^{(nm)} = & \mathbb{P} \Biggl\{\frac{p^{\textrm{V}}_{m,t} L^{\textrm{V}}_{m,t}}{ \gamma_{\textrm{V}} } (1 - \delta_{m,t}^2) | {e}_{m,t} |^2  \\ 
& \quad -  p^{\textrm{I}}_{n,t} L^{\textrm{I}}_{nm,t} {e}_{nm,t} \geq b_{m,t} \Bigm| {e}_{nm,t} \sim \hat{\mathcal{E}}_{nm}\Biggr\}, \\
\approx & \mathbb{P} \Biggl\{ c_t \left( | \hat{g}^{\textrm{I}}_{nm,t} |^2 + {e}_{nm,t}\right) \\
& \quad \leq \frac{\delta_{m}^2}{1 - \delta_{m}^2} | \hat{g}^{\textrm{V}}_{m,t} |^2  + | {e}_{m,t} |^2 \Bigm| {e}_{nm,t} \sim \hat{\mathcal{E}}_{nm} \Biggr\}.
\end{aligned}
\end{equation}
Moreover, the approximated probability in \eqref{beta_c_t}, defined as $\beta(c_t)$, is monotonically decreasing with respect to $c_t$.
\end{corollary}
Given the monotonicity shown in Proposition 1 and Corollary 2, we can efficiently derive the feasible region of $c_t$ in \eqref{opt6} by bisection search. Then, a one-dimensional search can be used to find the optimal $c_t$. The complete process for solving \eqref{opt6} of all matchings in the adaptation phase is given in Algorithm \ref{algo2}. Building on the \ac{PDF} estimation acquired in absorption, the \ac{RSU} can finally optimize the real-time transmit power schemes in adaptation to fulfill the \ac{QoS} requirements of the \ac{C-V2X} network. Together, these phases form a comprehensive framework for instilling resilience in C-V2X networks.
\begin{algorithm}[t]
\footnotesize
\caption{Transmit power scheme optimization in adaptation phase}
\label{algo2}
\begin{algorithmic}[1] 
\Require matching $\boldsymbol{A} $ and real-time imperfect \ac{CSI} $ \hat{\mathcal{H}}_t$
\Ensure Transmit power scheme $\boldsymbol{p}_{\mathcal{M},t}$ and $\boldsymbol{p}_{\mathcal{N},t}$
\For{\ac{V2V} link $m$ and \ac{V2I} link $n$ with $\alpha_{mn}=1$}
    \State Derive the feasible region of \eqref{c6-1}, \eqref{c6-2} and \eqref{c6-3} as $c_{l,t} \leq c_t \leq c_{u,t}$ through bisection search \;
    \State Find $c^{*}_t \in \left( c_{l,t}, c_{u,t}\right)$ that minimizes $| u(c^{*}_t) - 1 |$ through one dimensional search \;
    \If{$c^{*}_t \leq \frac{ \gamma_{\textrm{V}} p^{\textrm{I}}_{\text{min}} L^{\textrm{I}}_{nm,t}}{p^{\textrm{V}}_{\text{max}} L^{\textrm{V}}_{m,t} (1-\delta_m^2)}$}
        \State $p_{m,t}^{*} = p^{\textrm{V}}_{\text{max}}$ and $p_{n,t}^{*} =  p^{\textrm{I}}_{\text{min}}$ \;
    \ElsIf{$ \frac{ \gamma_{\textrm{V}} p^{\textrm{I}}_{\text{min}} L^{\textrm{I}}_{nm,t}}{p^{\textrm{V}}_{\text{max}} L^{\textrm{V}}_{m,t} (1-\delta_m^2)} <  c^{*}_t \leq \frac{ \gamma_{\textrm{V}} p^{\textrm{I}}_{\text{max}} L^{\textrm{I}}_{nm,t}}{p^{\textrm{V}}_{\text{max}} L^{\textrm{V}}_{m,t} (1-\delta_m^2)}$}
        \State $p_{m,t}^{*} = p^{\textrm{V}}_{\text{max}}$ and $p_{n,t}^{*} = \frac{c^{*}_t p^{\textrm{V}}_{\text{max}} L^{\textrm{V}}_{m,t} (1-\delta_m^2) }{\gamma_{\textrm{V}} L^{\textrm{I}}_{nm,t}} $ \;
    \Else
        \State $(p^{\textrm{V}}_{m,t})^{*} = \frac{ \gamma_{\textrm{V}} p^{\textrm{I}}_{\text{max}} L^{\textrm{I}}_{nm,t}}{c^{*}_t L^{\textrm{V}}_{m,t} (1-\delta_m^2)}$ and $(p^{\textrm{I}}_{n,t})^{*} = p^{\textrm{I}}_{\text{max}}$ \;
    \EndIf
\EndFor
\end{algorithmic}
\end{algorithm}

\vspace{-10pt}
\section{Simulation Results and Analysis}
For our simulations, we use a $400 \times 400 \ \text{m}^2$ area Manhattan mobility model with the \ac{RSU} deployed at the center of this area. The number of \ac{V2V} and \ac{V2I} links are $M = N = 10$. The distance between the transmitting vehicle and receiving vehicle of a \ac{V2V} link is chosen randomly between $60$~m and $80$~m according to a uniform distribution. The transmitting power requirements are $p_{\textrm{min}}^{\textrm{V}} = 10$~dBm, $p_{\textrm{max}}^{\textrm{V}} = 10$~dBm and $p_{\textrm{min}}^{\textrm{I}} = 10$~dBm, $p_{\textrm{max}}^{\textrm{I}} = 10$~dBm. The path loss exponents for all channel models are $\alpha = 3$ and the path loss models are detailed in Table \ref{channel simu}. The large-scale fading is assumed to vary every $\num{1000}$ time slots, i.e., $T = \num{1000}$. The noise spectrum density is $-174$ dBm/Hz. The other parameters are given by $B = 2$~MHz, $f_c = 5.9$~GHz, $\Delta_t = 1$~ms \cite{8993812}, $v = 10$~m/s, $D=\num{3200}$~bits \cite{9013252}, $\tau_0 = 15$~ms, $R_0 = 20$~Mbps, $P_0 = 95 \%$, and $K_1 = K_2 = 10$. All statistical results are averaged over \num{100000} channel realizations.
\begin{table}[t]
\centering
\caption{Path loss model}
\scriptsize
\vspace{-5pt}
\label{channel simu}
    \begin{tabular}{|c|c|c|}
    \hline
    \textbf{Channels}    & \textbf{Path loss}     & \textbf{\makecell[c]{Shadowing \\ standard deviation}}    \\ \hline
    $h^{\textrm{V}}_m$ & WINNER + B1 (LOS) \cite{kyosti2007winner} & $4$ dB \cite{9400748} \\ \hline
    $h^{\textrm{I}}_n$ & $128.1 + 37.6 \log_{10}d_n$ ($d_n$ in km)  & $8$ dB \cite{9400748} \\ \hline
    $h^{\textrm{I}}_{nm}$ & WINNER + B1 (NLOS)  \cite{kyosti2007winner}  & $8$ dB  \\ \hline
    $h^{\textrm{V}}_{mn}$ & WINNER + B1 (NLOS)           & $8$ dB  \\ \hline
    \end{tabular}
    \vspace{-0.5cm}
\end{table}

We first evaluate the effectiveness of the proposed absorption phase in estimating the \ac{PDF} of \ac{CSI} error distribution. We consider the case where parameter vector $\boldsymbol{\lambda}$ are equal for all $M$ \ac{V2V} links, i.e., $\lambda_1 = \ldots = \lambda_M = \lambda_\textrm{V} = 0.5$. Two different \ac{CSI} error distributions are assumed. Type I follows a \ac{GMM} with two components $x_1 \sim \mathcal{N}(0.2,0.04)$ and $x_2 \sim \mathcal{N}(0.8,0.02)$ of equal weight while Type II has different components $x_3 \sim \mathcal{N}(0.4,0.02)$ and $x_4 \sim \mathcal{N}(0.6,0.04)$ with different weight $0.4$ and $0.6$. Fig. \ref{pdf_two_type} shows the true \ac{PDF} of \ac{CSI} error distribution and the obtained \ac{PDF} on the worst match and best match obtained in sub-problem \eqref{opt2}. Specifically, the worst match and best match are the matching of the highest and lowest edge weight $\phi^{*}_{nm}$, indicating the worst and best adaptation capability respectively. From Fig. \ref{pdf_two_type}, we observe that the proposed absorption phase effectively enhances the adaptation capability of the \ac{C-V2X} network by accurately estimating the \ac{PDF} of the unknown error distribution, which is crucial for the subsequent adaptation. Furthermore, while the obtained \ac{PDF} for the worst match may not be deemed entirely precise, the resulting \ac{QoS} in the adaptation phase is still largely recovered, which will be later demonstrated in Fig. \ref{delay_two_phases}. For the subsequent simulations, we focus solely on Type I \ac{CSI} error distribution.
\begin{figure}[t]
\vspace{-6pt}
\captionsetup[subfigure]{font=small}  
    \centering
    \subfloat[]{
        \includegraphics[width=0.24\textwidth]{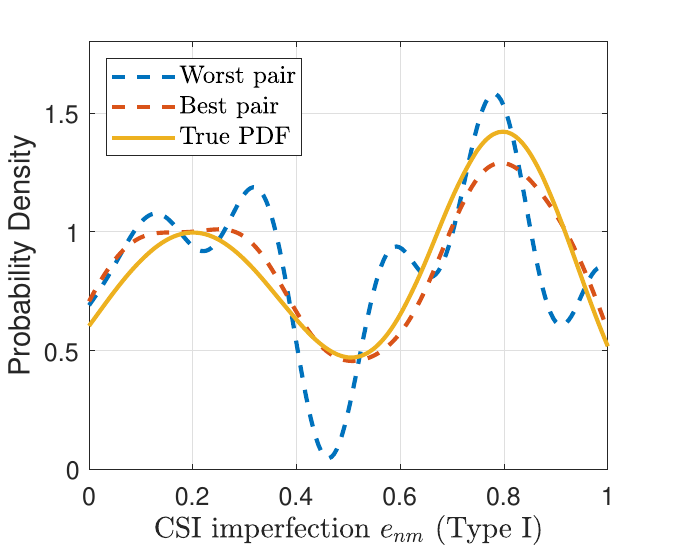}
        \label{pdf_1}
    }
    \subfloat[]{
        \includegraphics[width=0.24\textwidth]{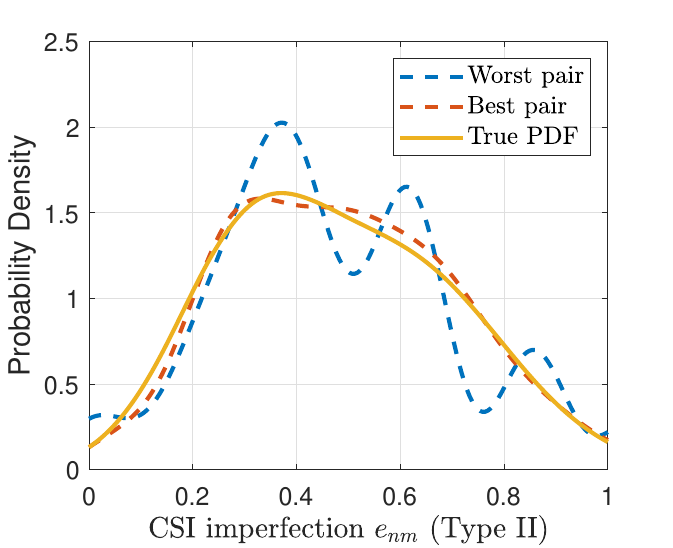}
        \label{pdf_2}
    }
    \vspace{-6pt}
    \caption{Deconvolution-based estimation obtained in the absorption phase: a) Type I error distribution, b) Type II error distribution. }
    \label{pdf_two_type}
        \vspace{-18pt}
\end{figure}

\begin{figure}[t]
\captionsetup[subfigure]{font=small}  
    \centering
    \subfloat[]{
        \includegraphics[width=0.24\textwidth]{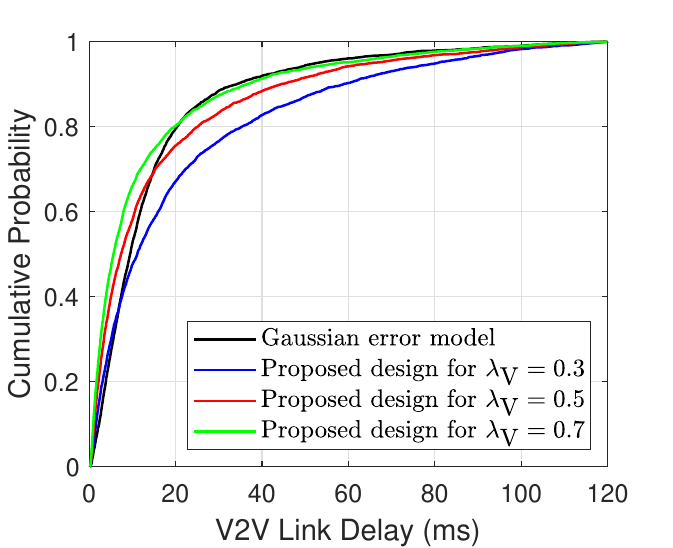}
        \label{plot_dalay_a}
    }
    \subfloat[]{
        \includegraphics[width=0.24\textwidth]{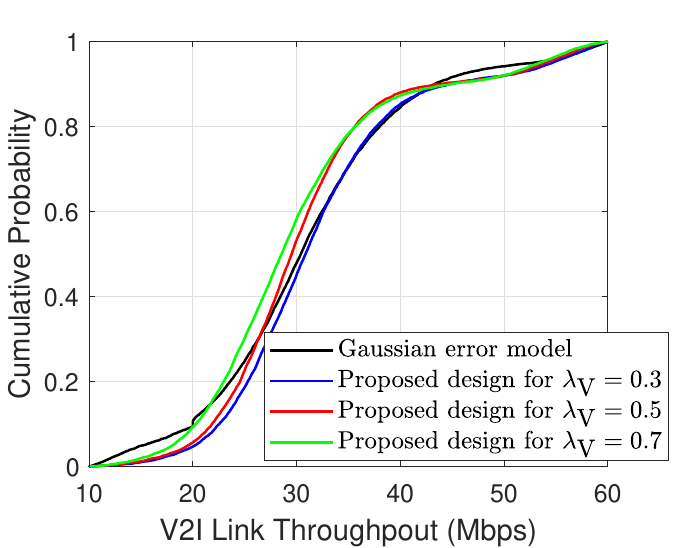}
        \label{fig:subfig2}
    }
        \vspace{-5pt}
    \caption{\ac{CDF} of \ac{QoS} on vehicular links in the absorption phase: a) Delay on \ac{V2V} links, b) Throughput on \ac{V2I} links.}
    \label{plot_absorption} 
            \vspace{-18pt}
\end{figure}
Fig. \ref{plot_absorption} shows the \ac{CDF} of \ac{QoS} on vehicular links during the absorption phase. The \ac{CDF} is obtained by considering all \ac{V2V} links and \ac{V2I} links. We compare with a Gaussian error model \cite{9400748} where the transmit power and matching are optimized based on the assumption that the \ac{CSI} error distribution is Gaussian. Fig. \ref{plot_absorption} shows that higher values of $\lambda_{\textrm{V}}$ leads to sacrificing the \ac{V2I} \ac{QoS} to compensate the \ac{V2V} \ac{QoS}. For instance, increasing $\lambda_{\textrm{V}}$ from $0.3$ to $0.5$ results in a $5$~Mbps  degradation in \ac{V2V} throughput, but improves the probability of satisfying the delay requirement by $14 \%$. The reason is that, under a higher $\lambda_{\textrm{V}}$, the \ac{HR} constraints in \eqref{c2-2} require the \ac{C-V2X} system to allocate more resources to \ac{V2V} links to ensure less degradation in delay. Therefore, the \ac{RSU} can modify $\lambda_{\textrm{V}}$ according to the different priority and criticality over the \ac{V2V} and \ac{V2I} links. In Fig. \ref{plot_dalay_a}, we can further observe that the probability of meeting the delay requirement on \ac{V2V} links, in case $\lambda_{\textrm{V}} = 0.5$, is only  $70 \%$, which not only falls significantly short of the required threshold $P_0 = 95 \%$, but is also $1 \%$ lower than that achieved by the Gaussian error model benchmark. However, this performance gap reflects a deliberate tradeoff inherent in resilience design. Specifically, the absorption prioritizes accurate \ac{PDF} estimation over immediate \ac{QoS} satisfaction, thereby laying the foundation for more effective \ac{QoS} recovery in the subsequent adaptation.

Next, we evaluate the \ac{QoS} of vehicular links during the adaptation phase. We focus on the case $\lambda_{\textrm{V}} = 0.5$. For comparison, we consider an alternative benchmark approach based on \ac{HPR} proposed in \cite{9382930}. The \ac{HPR}-based design assumes a Gaussian error distribution in the absorption phase for resource allocation and constructs a region that encompasses $P_0$ of the collected imperfect \ac{CSI} samples. The \ac{HPR} is then used to design a transmit power allocation scheme during adaptation for satisfying the delay requirement with a probability of at least $P_0$. As illustrated in Fig. \ref{delay_adaptation}, both the proposed design and the \ac{HPR}-based approach outperform the Gaussian error model on the delay of \ac{V2V} links during the adaptation phase. This is because both designs exploit the statistical characteristics of the error distribution during the absorption phase. Moreover, the \ac{HPR}-based design achieves a probability of $P_0 = 93 \%$ in meeting the delay requirement and the proposed design attains a slightly lower (by about $1 \%$) probability. However, the proposed design performs better in mitigating severe delay degradation, as shown in the \ac{CCDF} of Fig. \ref{ccdf_delay_adaptation}. In particular, for cases in which the delay exceeds $\tau_0 = 15$~ms, the proposed design ensures a $92 \%$ probability that the delay remains below $40$~ms, representing a $22 \%$ and $30 \%$ improvement over the Gaussian error model and the \ac{HPR}-based design, respectively. Such advantage stems from the proposed design’s utilization of the full \ac{PDF} of the error distribution, which provides richer statistical information compared to the \ac{HPR}. These results highlight the superior resilience of the proposed design in mitigating severe delay degradation after absorption. 

We further analyze the instantaneous probability of satisfying delay \ac{QoS} across \num{1000} time slots on all \ac{V2V} links. The benefits of full \ac{PDF} estimation in our proposed design are evident in Fig. \ref{delay_prob_adaptation}, where the proposed design consistently ensures that the actual probability does not deviate significantly from the requirement $P_0$. Additionally, the proposed design achieves significantly higher throughput on \ac{V2I} links compared to the two benchmarks, as shown in Fig. \ref{throughput_adaptation}. For instance, at a cumulative probability of $0.6$, the proposed design yields approximately $16 \%$ higher throughput than the two benchmarks. In addition, the proposed design achieves an average throughput gain of $14 \%$ and $16 \%$ over the Gaussian error model and the \ac{HPR}-based design. This is attributed to the ability of the full \ac{PDF} to facilitate a more flexible and precise transmit power allocation strategy. In contrast, the \ac{RSU} in the \ac{HPR}-based design must adopt a conservative worst-case probability approach, while the Gaussian error model is affected by inaccuracies resulted from an incorrect assumption about the \ac{CSI} error distribution.
\begin{figure}[t]
\vspace{-6pt}
\captionsetup[subfigure]{font=small}  
    \centering
    \subfloat[]{
        \includegraphics[width=0.24\textwidth]{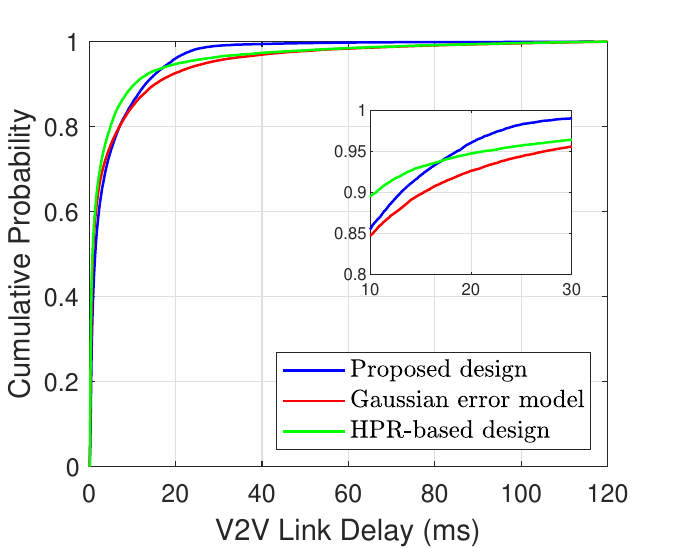}
        \label{delay_adaptation}
    }
    \subfloat[]{
        \includegraphics[width=0.24\textwidth]{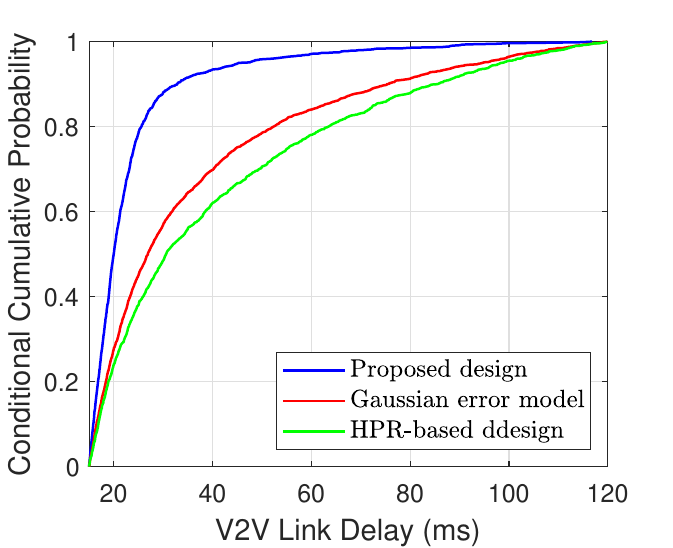}
        \label{ccdf_delay_adaptation}
    }
    \\ 
    \vspace{-11pt}
    \subfloat[]{
        \includegraphics[width=0.24\textwidth]{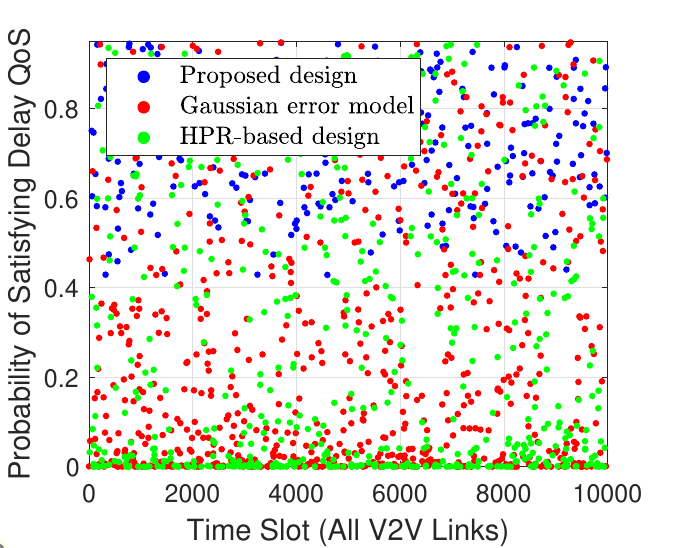}
        \label{delay_prob_adaptation}
    }
    \subfloat[]{
        \includegraphics[width=0.24\textwidth]{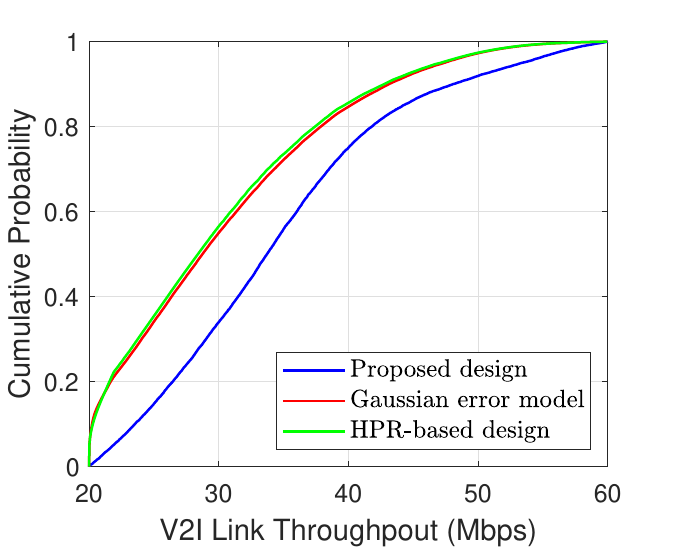}
        \label{throughput_adaptation}
    }
        \vspace{-4pt}
    \caption{\ac{QoS} on vehicular links in the adaptation phase with different design: a) \ac{CDF} of delay on \ac{V2V} links, b) \ac{CCDF} of delay on \ac{V2V} links, c) The true probability of satisfying delay requirement on \ac{V2V} links, d) \ac{CDF} of throughput on \ac{V2I} links.}
    \label{fig:example}
        \vspace{-20pt}
\end{figure}

Fig. \ref{delay_two_phases} shows the delay of \ac{V2V} links across both the absorption and adaptation phases. Specifically, we focus on the \ac{V2V} link exhibiting the worst adaptation capability, i.e., the \ac{V2V} link within the worst match given in Fig. \ref{pdf_1}. An absorption phase consisting of $T = \num{1000}$ time slots is considered, followed by a adaptation phase of \num{200} time slots. To analyze system performance, we sample \num{200} time slots from the absorption phase at equal intervals and include all \num{200} time slots from the adaptation phase, resulting in a total of \num{400} time slots. In this context, the transition from absorption to adaptation occurs at the \num{200}-th time slot. In the absorption phase, the proposed design limits the peak delay to $120$~ms, while the Gaussian error model and \ac{HPR}-based design exhibit peaks exceeding $140$~ms. This performance gain is attributed to the incorporation of the \ac{HR} metric into the optimization process, which mitigates transient delay spikes under imperfect \ac{CSI} disruption. Furthermore, during the subsequent adaptation phase, the proposed design maintains a lower and more stable delay on the \ac{V2V} link due to the accurate estimated \ac{PDF} of the \ac{CSI} error distribution. Specifically, when considering only delay instances exceeding $\tau_0 = 15$~ms, it achieves a conditional mean delay of $19.5$~ms, which is $35 \%$ and $56 \%$ lower than that achieved the Gaussian error model and \ac{HPR}-based design, respectively.
\begin{figure*}[t]
    \vspace{-15pt}
\captionsetup[subfigure]{font=small}  
    \centering
    \subfloat[]{
        \includegraphics[width=0.28\textwidth]{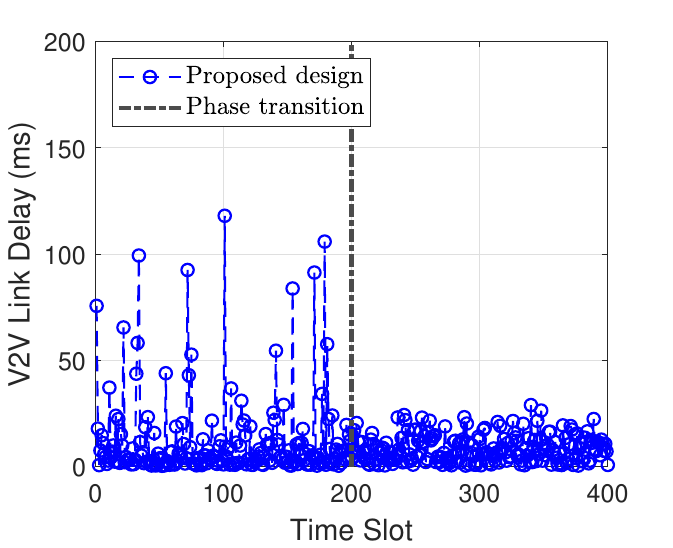}
        \label{delay_two_phases_1}
    }
    \subfloat[]{
        \includegraphics[width=0.28\textwidth]{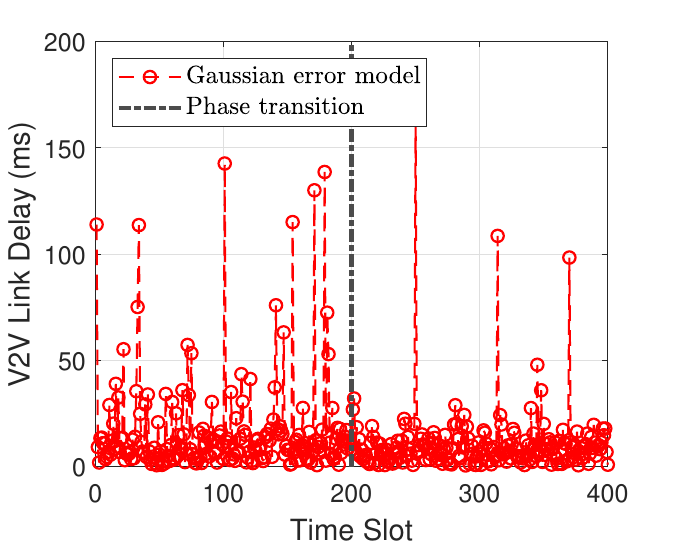}
        \label{delay_two_phases_2}
    }
    \subfloat[]{
        \includegraphics[width=0.28\textwidth]{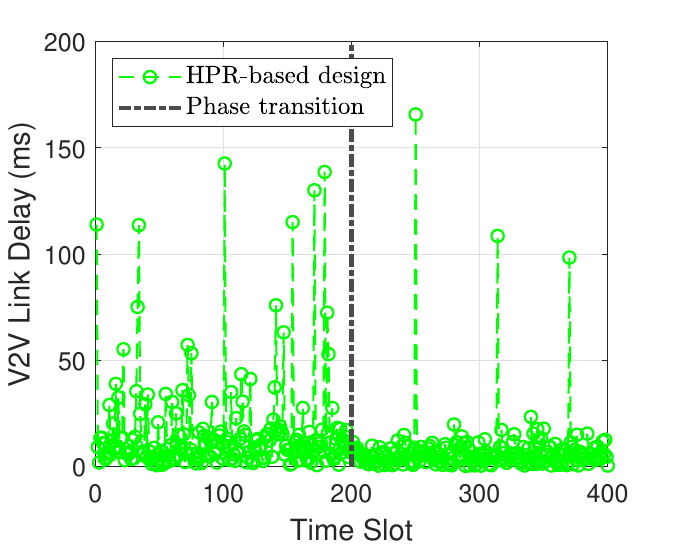}
        \label{delay_two_phases_3}
    }
    \vspace{-3pt}
    \caption{Delay of the \ac{V2V} link with the worst adaptation capability in two phases.}
    \vspace{-18pt}
        \label{delay_two_phases}
\end{figure*}

\vspace{-6pt}
\section{Conclusion}
\vspace{-4pt}
In this paper, we have proposed a novel two-phase framework that instills resilience into \ac{C-V2X} network under imperfect \ac{CSI}. Specifically, we have formulated a bi-level optimization problem aiming at satisfying the \ac{QoS} requirements on vehicular links without any prior assumptions on the imperfect \ac{CSI}. Then, we have decoupled the complex bi-level optimization problem into two sub-problems that are solved sequentially in the \emph{absorption phase} and \emph{adaptation phase}. For the absorption phase, we have defined the \ac{MSE} of imperfect \ac{CSI}'s \ac{PDF} as the adaptation capability and derived an explicit upper bound on the \ac{MSE}. Based on the analytic expression of the adaptation capability, the matching and absorption power scheme have been optimized. Due to the tradeoff between the adaptation capability and the \ac{C-V2X}'s \ac{QoS} during absorption, we have incorporated a novel metric named \ac{HR} to evaluate the \ac{C-V2X}'s absorption performance. After the absorption phase, we have further analyzed the impact of the absorption phase on the adaptation phase and optimized the real-time transmit power scheme based on the estimated \ac{PDF}. The simulation results demonstrate the superior capability of the proposed framework in sustaining the \ac{QoS} of the \ac{C-V2X} network under imperfect \ac{CSI}, across both two phases. Specifically, the proposed design reduces the conditional \ac{V2V} delay, delay values exceeding the desired requirement, by $35 \%$ and $56 \%$ and improves the average \ac{V2I} throughput by $14 \%$ and $16 \%$ over the benchmarks during adaptation, without compromising the \ac{C-V2X}'s \ac{QoS} in absorption.
\appendices
\vspace{-12pt}
\section{Proof of Lemma 1}
\label{Proof of Lemma 1}
\begin{proof}
According to the definition in \eqref{hr v2v}, we have
\begin{equation}
    \begin{aligned}
        \Lambda_m =  \lim_{\Delta \tau \rightarrow 0} \frac{\mathbb{P}\left\{ \gamma^{\textrm{V}}_m \geq 2^{\frac{D}{B(\tau_0+\Delta\tau)}} - 1 \right\} - \mathbb{P}\left\{ \gamma^{\textrm{V}}_m \geq 2^{\frac{D}{B\tau_0}} - 1 \right\}}{\Delta \tau \left( 1 - \mathbb{P}\left\{ \gamma^{\textrm{V}}_m \geq 2^{\frac{D}{B\tau_0}} - 1 \right\} \right)}.
    \end{aligned}
\end{equation}
We then focus on the expression of the following probability:
\begin{equation}
\begin{aligned}
    \mathbb{P}\left\{ \gamma^{\textrm{V}}_m \geq 2^{\frac{D}{B\tau_0}} - 1 \right\} 
    = \frac{\exp{\left[-\frac{\sigma_n^2}{p^{\textrm{V}}_{m, \text{a}} L^{\textrm{V}}_{m,\text{a}}}(2^{\frac{D}{B\tau_0}} - 1)\right]}}{1 + \frac{p^{\textrm{I}}_{n,\text{a}} L^{\textrm{I}}_{n,\text{a}}}{p^{\textrm{V}}_{m,\text{a}} L^{\textrm{V}}_{m,\text{a}}}(2^{\frac{D}{B\tau_0}} - 1)}.
\end{aligned}
\end{equation}
Similarly, we have:
\begin{equation}
\small
    \mathbb{P}\left\{ \gamma^{\textrm{V}}_m \geq 2^{\frac{D}{B(\tau_0+\Delta\tau)}} - 1 \right\} = \frac{\exp{\left[-\frac{\sigma_n^2}{p^{\textrm{V}}_{m,\text{a}} L^{\textrm{V}}_{m,\text{a}}}(2^{\frac{D}{B(\tau_0+\Delta \tau)}} - 1)\right]}}{1 + \frac{p^{\textrm{I}}_{n,\text{a}} L^{\textrm{I}}_{n,\text{a}}}{p^{\textrm{V}}_{m,\text{a}} L^{\textrm{V}}_{m,\text{a}}}(2^{\frac{D}{B(\tau_0+\Delta \tau)}} - 1)}.
\end{equation}
Then, we can rewrite the \eqref{hr v2v} as following
\begin{equation}
    \Lambda_m = \lim_{\Delta \tau \rightarrow 0} \frac{k(\Delta \tau) - A}{\Delta \tau(1 - A)},
\end{equation}
where
\begin{equation}
    k(\Delta \tau) = \frac{\exp{\left[-\frac{\sigma_n^2}{p^{\textrm{V}}_{m, \text{a}} L^{\textrm{V}}_{m,\text{a}}}(2^{\frac{D}{B(\tau_0+\Delta \tau)}} - 1)\right]}}{1 + \frac{p^{\textrm{I}}_{n,\text{a}} L^{\textrm{I}}_{n,\text{a}}}{p^{\textrm{V}}_{m,\text{a}} L^{\textrm{V}}_{m,\text{a}}}(2^{\frac{D}{B(\tau_0+\Delta \tau)}} - 1)}, 
\end{equation}
\vspace{-5pt}
\begin{equation}
    A =  \frac{\exp{\left[-\frac{\sigma_n^2}{p^{\textrm{V}}_{m,\text{a}} L^{\textrm{V}}_{m,\text{a}}}(2^{\frac{D}{B\tau_0}} - 1)\right]}}{1 + \frac{p^{\textrm{I}}_{n, \text{a}} L^{\textrm{I}}_{n, \text{a}}}{p^{\textrm{V}}_{m,\text{a}} L^{\textrm{V}}_{m,\text{a}}}(2^{\frac{D}{B\tau_0}} - 1)}.
\end{equation}
Moreover, we can observe that $\lim_{\Delta \tau \rightarrow 0} k(\Delta \tau) = A$. Thus, we can derive $\Lambda_m$ by L'Hôpital's rule, which is given by:
\begin{equation}
\begin{aligned}
    \Lambda_m
    & = D_{\textrm{V}} e^{-\frac{\sigma^2\gamma_{\textrm{V}}}{p^{\textrm{V}}_{m,\text{a}}L^{\textrm{V}}_{m,\text{a}}}} \frac{ \frac{p^{\textrm{I}}_{n,\text{a}}L^{\textrm{I}}_{nm, \text{a}}}{p^{\textrm{V}}_{m,\text{a}}L^{\textrm{V}}_{m,\text{a}}} + \frac{\sigma^2}{p^{\textrm{V}}_{m, \text{a}} L^{\textrm{V}}_{m,\text{a}}} \left( 1 + \frac{p^{\textrm{I}}_{n,\text{a}}L^{\textrm{I}}_{nm,\text{a}}}{p^{\textrm{V}}_{m,\text{a}}L^{\textrm{V}}_{m,\text{a}}} \gamma_{\textrm{V}} \right) }{ \left( 1 + \frac{p^{\textrm{I}}_{n,\text{a}}L^{\textrm{I}}_{nm,\text{a}}}{p^{\textrm{V}}_{m,\text{a}}L^{\textrm{V}}_{m,\text{a}}} \gamma_{\textrm{V}} - e^{-\frac{\sigma^2\gamma_{\textrm{V}}}{p^{\textrm{V}}_{m,\text{a}}L^{\textrm{V}}_{m,\text{a}}}} \right)^2 },
\end{aligned}
\end{equation} 
where $D_{\textrm{V}} = \frac{\ln{2}D2^{\frac{D}{B\tau_0}}}{B\tau_0^2} $.
\end{proof}
\vspace{-12pt}
\section{Proof of Theorem 2}
\label{Proof of Theorem 2}
\begin{proof}
First, we define $\theta_t = \frac{1}{2\pi} \int_{-\infty}^{\infty} F\left\{ \Phi_t \right\} F^{*}\left\{ f_Z \right\} (1-\frac{jw}{\lambda_Y}) dw$ and $\hat{\theta}_t = \frac{1}{2\pi T}\sum_{k=1}^T \int_{-K \pi}^{K \pi} F\left\{ \Phi_t \right\} e^{jwz_k}(1-\frac{jw}{\lambda_Y}) dw$. Then we have $\mathbb{E}\left[ \left(\hat{P}_{m,t}^{(nm)} - P_{m,t}^{(nm)} \right)^2 \right] = \mathbb{E} \left[ \left( \theta_t - \hat{\theta}_t \right)^2\right] = \left( \theta_t - \mathbb{E} \left[ \hat{\theta}_t \right] \right)^2 + \text{Var} \left[ \hat{\theta}_t \right]$. Therefore, we can derive
\begin{equation}
\small
\begin{aligned}
      \theta_t - \mathbb{E} \left[ \hat{\theta}_t \right] 
    = &  \frac{1}{2\pi} \int_{-\infty}^{\infty} F\left\{ \Phi_t \right\} F^{*}\left\{ f_Z \right\} (1-\frac{jw}{\lambda_Y}) dw \\
    & - \frac{1}{2\pi T}\sum_{k=1}^T \int_{-K \pi}^{K \pi} F\left\{ \Phi \right\} \mathbb{E} \left[ e^{jwz_k} \right] (1-\frac{jw}{\lambda_Y}) dw \\
    = & \frac{1}{2\pi} \int_{-\infty}^{\infty} F\left\{ \Phi_t \right\} F^{*}\left\{ f_Z \right\} (1-\frac{jw}{\lambda_Y}) dw \\
    & - \frac{1}{2\pi} \int_{-K \pi}^{K \pi} F\left\{ \Phi_t \right\} F^{*}\left\{ f_Z \right\} (1-\frac{jw}{\lambda_Y}) dw  \\
    = & \frac{1}{2\pi} \int_{w\geq |K\pi|} F\left\{ \Phi_t \right\} F^{*}\left\{ f_E \right\} dw.
\end{aligned}
\end{equation}
Moreover, we can obtain \eqref{MSE SUM} based on \eqref{inequality_V_E}.
\begin{figure*}[b]
\small
    \vspace{-12pt}
    \begin{equation}
\label{inequality_V_E}
\begin{aligned}
        \text{Var} \left[ \hat{\theta}_t \right] \leq \mathbb{E} \left[ \hat{\theta}_t^2 \right]  \leq \frac{1}{4\pi^2T} E\left\{ \left( \int_{- K\pi}^{K\pi} \left\| F\left\{ \Phi_t \right\} \right\|  \left\| 1-\frac{jw}{\lambda_Y} \right\| \left\|e^{jwZ}\right\| dw \right)^2 \right\}.
\end{aligned}
\end{equation}
    \vspace{-8pt}
    \begin{equation}
    \label{MSE SUM}
    \begin{aligned}
 \mathbb{E}\left[ \left(\hat{P}_{m,t}^{(nm)} - P_{m,t}^{(nm)} \right)^2 \right] \leq \left( \frac{1}{2\pi} \int_{w\geq |K\pi|} F\left\{ \Phi_t \right\} F^{*}\left\{ f_E \right\} dw \right)^2 + \frac{1}{4\pi^2T} E\left\{ \left( \int_{- K\pi}^{K\pi} \left\| F\left\{ \Phi_t \right\} \right\|  \left\| 1-\frac{jw}{\lambda_Y} \right\| \left\|e^{jwZ}\right\| dw \right)^2 \right\}.
    \end{aligned}
    \end{equation}
    \vspace{-20pt}
\end{figure*}
To derive an analytic expression of the second term on the right side of \eqref{MSE SUM}, we have
\begin{equation}
\small
\begin{aligned}
    \left\| F\left\{ \Phi_t \right\} \right\| 
    & = \left\| \frac{e^{-jwK_1} - 1}{-jw} + \frac{e^{-(c_t+jw)K_1 }- 1}{c_t+jw} \right\| \\
    & \overset{(a)}{\approx} \left\| \frac{e^{-jwK_1} - 1}{-jw} + \frac{- 1}{c_t+jw} \right\| \\
    & = \left\|  \frac{e^{-jwK_1}}{-jw} + \frac{c_t}{jw(c_T+jw)}  \right\| \\
    & \overset{(b)}{\leq}  \left\|  \frac{1}{jw} \right\|  + \left\| \frac{c_t}{jw(c_t+jw)}  \right\|, 
\end{aligned}
\end{equation}
where the approximation (a) is based on the fact that $K_1$ is a large constant such that $e^{-(c_t+jw)K_1} \approx 0$ and (b) is based on triangle inequality. Therefore, we can obtain
\begin{equation}
\small
\label{47}
\begin{aligned}
    & E\left\{ \left( \int_{- K\pi}^{K\pi} \left\| F\left\{ \Phi_t \right\} \right\|  \left\| 1-\frac{jw}{\lambda_Y} \right\| \left\|e^{jwZ}\right\| dw \right)^2 \right\} \\
    \leq & \left[ \int_{- K\pi}^{K\pi} \left( \left\|  \frac{1}{jw} \right\|  + \left\| \frac{c_t}{jw(c_t+jw)}  \right\|  \right) \sqrt{1 + (\lambda_{Y}^{-1}w)^2 } dw \right]^2 \\
    = & 4 \left[ \int_{0}^{K\pi} \left( \frac{1}{w} + \frac{c_t}{w\sqrt{c_t^2 + w^2}} \right) \sqrt{1 + (\lambda_{Y}^{-1}w)^2 } dw \right]^2. 
\end{aligned}
\end{equation}
Finally, we can obtain an upper bound of \eqref{47}:
\begin{equation}
\small
\label{MSE integral}
    \begin{aligned}
        & \int_{0}^{K\pi} \left( \frac{\sqrt{1 + (\lambda_{Y}^{-1}w)^2 }}{w} + \frac{c_t\sqrt{1 + (\lambda_{Y}^{-1}w)^2 }}{w\sqrt{c_t^2 + w^2}} \right) dw \\
        \leq & \int_{0}^{K\pi} \left( \frac{\sqrt{1 + (\lambda_{Y}^{-1}w)^2 }}{w} + \frac{c_t\sqrt{(1 + \lambda_{Y}^{-1}w)^2 }}{w\sqrt{c_t^2 + w^2}} \right) dw \\
        = & \sqrt{1+\lambda_{Y}^{-2}K^2\pi^2} + \ln{\left(\frac{\sqrt{1+\lambda_{Y}^{-2}K^2\pi^2} - 1}{\lambda_{Y}^{-1}K\pi}\right)} - 1 \\
        & + c_t\lambda_{Y}^{-1} \ln{ \left(\frac{ K\pi + \sqrt{c_t^2 + K^2\pi^2}}{c_t}\right)} \\
        & + \frac{1}{c_t} \ln{\left( \frac{c_t K\pi}{\sqrt{K^2\pi^2 + c_t^2} + K\pi} \right)}.
    \end{aligned}
\end{equation}
By substituting \eqref{MSE integral} into \eqref{MSE SUM}, we complete the proof.
\end{proof}
\section{Proof of Proposition 1}
\label{Proof of Proposition 1}
\begin{proof}
The monotonicity of $u(c_t)$ is determined by last two terms in \eqref{u}, which can be defined and rewritten as
\begin{equation}
\label{hat u}
\begin{aligned}
\hat{u}(c_t) 
\triangleq & \frac{c_t}{\lambda_{Y}} \ln{\left( K_2\pi + \sqrt{c_t^2 + K_2^2\pi^2} \right)} - \frac{c_t}{\lambda_{Y}} \ln{c_t} + \frac{1}{c_t} \ln{(K_2\pi)} \\ 
& + \frac{1}{c_t}\ln{c_t} - \frac{1}{c_t} \ln{\left( K_2\pi \right)} - \frac{1}{c_t} \ln{\left( \sqrt{1 + (\frac{c_t}{K_2\pi})^2} + 1 \right)} \\
= & \frac{c_t}{\lambda_{Y}} \ln{\left( K_2\pi + \sqrt{c_t^2 + K_2^2\pi^2} \right)} + \frac{1}{c_t}\ln{c_t} - \frac{c_t}{\lambda_{Y}} \ln{c_t} \\
& - \frac{1}{c_t} \ln{\left( \sqrt{1 + (\frac{c_t}{K_2\pi})^2} + 1 \right)} \\
= & \ln{\left( c_t^{-1} + \sqrt{\left(c_t\right)^{-2} + \left(K_2\pi\right)^{-2}} \right)} \left(\frac{c_t}{\lambda_{Y}} - c_t^{-1}\right) \\
& + \frac{c_t}{\lambda_{Y}} \ln{(K_2\pi)} \\
\overset{(a)}{=} & -\frac{1}{x \lambda_{Y}} \ln{D} + \ln{\left( x + \sqrt{x^2 +D^{2}} \right)} \left(\frac{1}{x\lambda_{Y}} - x\right),
\end{aligned}
\end{equation}
where we define $D \triangleq (K_2\pi)^{-1}$ and $x \triangleq c_t^{-1}$ in (a). The first derivative of $\hat{u}(c_t) = \hat{u}(x)$ with respect to $x$ is
\begin{equation}
\label{derivative}
\begin{aligned}
&\frac{\mathrm{d}\hat{u}(x)}{\mathrm{d}x} \\
= & \frac{ x + \ln{D}\sqrt{x^2 + D^{2}} - \sqrt{x^2 + D^{2}} \ln{\left( x + \sqrt{x^2 + D^{2}} \right)}}{ \lambda_Y x^2 \sqrt{x^2 + D^{2}}} \\
& - \frac{ \lambda_Y x^2 \sqrt{x^2 + D^{2}} \ln{\left( x + \sqrt{x^2 + D^{2}} \right)}}{ \lambda_Y x^2 \sqrt{x^2 + D^{2}}} - \frac{x}{\sqrt{x^2 + D^2}}.
\end{aligned} 
\end{equation}
According to \eqref{derivative}, one sufficient condition for $\frac{\mathrm{d}\hat{u}(x)}{\mathrm{d}x} \leq 0$ is
\begin{equation}
\begin{aligned}
    \sqrt{x^2 + D^{2}} \Biggl[ & \ln{D} - (1+\lambda_Yx^2) \ln{\left( x + \sqrt{x^2 + D^{2}} \right)} \\
& + \frac{x}{\sqrt{x^2 + D^{2}}} \Biggr] \leq 0,
\end{aligned}
\end{equation}
which is equivalent to 
\begin{equation}
\label{con}
\begin{aligned}
    & \frac{x}{\sqrt{x^2 + D^{2}}} - \ln{\left(\frac{x + \sqrt{x^2 + D^{2}}}{D}\right)} \\
    &- \lambda_Yx^2 \ln{\left( x + \sqrt{x^2 + D^{2}} \right)} \leq 0.
\end{aligned}
\end{equation}
To satisfy \eqref{con}, we observe that as long as $K_2$ is large enough, i.e., $D$ is close to $0$, the term $-\ln{\left(\frac{x + \sqrt{x^2 + D^{2}}}{D}\right)} < 0$ is dominant in \eqref{con}. Thus, \eqref{con} can be always satisfied and $u(c_t)$ can be viewed as an increasing function of $c_t$.
\end{proof}
\bibliographystyle{IEEEtran}
\bibliography{bibliography}
\end{document}